\documentclass[12pt]{article}
\usepackage{amsfonts}
\usepackage{amsmath}
\usepackage{makeidx}
\usepackage{amssymb}
\usepackage{graphicx}
\usepackage[T1]{fontenc}

\newtheorem{theorem}{Theorem}

\newtheorem{corollary}[theorem]{Corollary}

\newtheorem{adefinition}[theorem]{Definition}

\newtheorem{lemma}[theorem]{Lemma}

\newtheorem{aremark}[theorem]{Remark}
\newenvironment{remark}{\begin{aremark}\rm}{\end{aremark}}

\newenvironment{proof}[1][Proof]{\noindent\textbf{#1.} }{\ \rule{0.5em}{0.5em}}
\numberwithin{equation}{section} \numberwithin{theorem}{section}

\begin{document}

\title{On Random Matrices Arising \\
in Deep Neural Networks: Gaussian Case}
\date{}
\author{{\small L. Pastur} \\
{\small B.Verkin Institute }\\
{\small for Low Temperature Physics and Engineering} \\
{\small Kharkiv, Ukraine}}

\maketitle

\begin{abstract}
The paper deals with distribution of singular values of product of
random matrices arising in the analysis of deep neural networks.
The matrices resemble the product analogs of the sample covariance
matrices, however, an important difference is that the population
covariance matrices, which are assumed to be non-random in the
standard setting of statistics and random matrix theory, are now
random, moreover, are certain functions of random data matrices.
The problem has been considered in recent work \cite{Pe-Co:18} by
using the techniques of free probability theory. Since, however,
free probability theory deals with population matrices which are
independent of the data matrices, its applicability in this case
requires an additional justification. We present this
justification by using a version of the standard techniques of
random matrix theory under the assumption that the entries of data
matrices are independent Gaussian random variables. In the
subsequent paper \cite{Pa-Sl:20} we extend our results to the case
where the entries of data matrices are just independent
identically distributed random variables with several finite
moments. This, in particular, extends the property of the
so-called macroscopic universality on the considered random
matrices.
\end{abstract}

\maketitle


\section{Introduction}

Deep neural networks with multiple hidden layers have achieved remarkable
performance in a wide variety of domains, see e.g. \cite%
{Be-Co:12,Bu:17,Ca-Ch:18,Le-Co:15,Sc:15,Sh-Ma:19} for reviews. Among
numerous research directions of the field those using random matrices of
large size are of considerable amount and interest. They treat random
untrained networks (allowing for the study their initialization and learning
dynamics, the information propagation through generic deep random neural
networks, etc.), the expressivity and the geometry of neural networks, the
analysis of the Bayesian approach, etc., see e.g. \cite%
{Gi-Co:16,Li-Qi:00,Ma-Ma:17,Ma-Co:16,Pe-Co:18,Po-Co:16,Sa-Co:11,Sc-Wa:17,Sc-Co:17}
and references therein.

Consider an untrained, feed-forward, fully connected neural network with $L$
layers of width $n_{l}$ for the $l$th layer and pointwise nonlinearities $%
\varphi $. Let%
\begin{equation}
x^{0}=\{x_{j_{0}}^{0}\}_{j_{0}=1}^{n_{0}}\in \mathbb{R}^{n_{0}}  \label{x0}
\end{equation}%
be the input to the network, and $x^{L}=\{x_{j_{L}}^{L}\}_{j_{L}=1}^{n_{L}}%
\in \mathbb{R}^{n_{L}}$ be its output. The components of the activations $%
x^{l}$ in the $l$th layer and the post-affine transformations $y^{l}$ of $%
x^{l}$ are $\{x_{j_{l}}^{l}\}_{j_{l}=1}^{n_{l}}$ and $\{y_{j_{l}}^{l}%
\}_{j_{l}=1}^{n_{l}}$ respectively and are related as follows
\begin{equation}
y^{l}=W^{l}x^{l-1}+b^{l},\;x_{j_{l}}^{l}=\varphi
(y_{j_{l}}^{l}),\;j_{l}=1,...,n_{l},\;l=1,...,L,  \label{rec}
\end{equation}%
where
\begin{equation}
W^{l}=\{W_{j_{l}j_{l-1}}^{l}\}_{j_{l},j_{l-1}=1}^{n_{l},n_{l-1}},\;l=1,...,L
\label{wl}
\end{equation}%
are $n_{l}\times n_{l-1}$ rectangular weight matrices,%
\begin{equation}
b^{l}=\{b_{j_{l}}^{l}\}_{j_{l}=1}^{n_{l}},\;l=1,2,...,L  \label{bl}
\end{equation}%
are $n_{l}$-component bias vectors and $\varphi :\mathbb{R}\rightarrow
\mathbb{R}$ is the component-wise nonlinearity.

Assume that the biases components $\{b_{j_{l}}^{l}\}_{j_{l}=1}^{n_{l}}$ are
the Gaussian random variables such that:
\begin{equation}
\mathbf{E}\{b_{j_{l}}^{l}\}=0,\;\mathbf{E}%
\{b_{j_{l_{1}}}^{l_{1}}b_{j_{l_{2}}}^{l_{2}}\}=\sigma _{b}^{2}\delta
_{l_{1}l_{2}}\delta _{j_{l_{1}}j_{l_{2}}}.  \label{bga}
\end{equation}%
As for the weight matrices $W^{l},\;l=1,2,...,L$, it is assumed that
\begin{eqnarray}
W^{l}
&=&n_{l-1}^{-1/2}X^{l}=n_{l-1}^{-1/2}\{X_{j_{l}j_{l-1}}^{l}%
\}_{j_{l},j_{l-1}=1}^{n_{l},n_{l-1}},  \label{wga} \\
\mathbf{E}\{X_{j_{l}j_{l-1}}^{l}\} &=&0,\;\mathbf{E}%
\{X_{j_{l_{1}}j_{l_{1}-1}}^{l_{1}}X_{j_{l_{2}}j_{l_{2}-1}}^{l_{2}}\}=\delta
_{l_{1}l_{2}}\delta _{j_{l_{1}}j_{l_{1}-1}}\delta _{j_{l_{2}}j_{l_{2}-1}},
\notag
\end{eqnarray}%
the matrices $X^{l},\;l=1,2,...,L$ are independent and identically
distributed and for every $l$ we view $X^{l}$ as the upper left rectangular
block of the semi-infinite random matrix%
\begin{equation}
\{X_{j_{l}j_{l-1}}^{l}\}_{j_{l},j_{l-1}=1}^{\infty ,\infty }  \label{xinf}
\end{equation}%
with the standard Gaussian entries.

Likewise, for every $l$ we view $b^{l}$ in (\ref{bl}) as the first $n_{l}$
components of the semi-infinite vector
\begin{equation}
\{b_{j_{l}}^{l}\}_{j_{l}=1}^{\infty }  \label{binf}
\end{equation}%
whose components are Gaussian random variables normalized by (\ref{bga})
with $n_{l}=\infty ,\;l=1,2,...,L$.

As a result of this form of weights and biases of the $l$th layer they are
for all $n_{l}=1,2,...$ defined on the same infinite-dimensional product
probability space $\Omega ^{l}$ generated by (\ref{xinf}) -- (\ref{binf}).
Let also
\begin{equation}  \label{oml}
\Omega _{l}=\Omega ^{l}\times \Omega ^{l-1}\times ...\times \Omega
^{1},\;l=1,...,L
\end{equation}%
be the infinite-dimensional probability space on which the recurrence (\ref%
{rec}) is defined for a given $l$ (the number of layers). This will allow us
to formulate our results on the large size asymptotics of the eigenvalue
distribution of matrices (\ref{JJM}) as those valid with probability 1 in $%
\Omega _{L}$.

Note that matrices $W^{l}(W^{l})^{T}$ of (\ref{wl}) and (\ref{wga}) are
known in statistics as the Wishart matrices \cite{Mu:05}.

Consider the input-output Jacobian
\begin{equation}
J_{\mathbf{n}_{L}}^{L}:=\left\{ \frac{\partial x_{j_{L}}^{L}}{\partial
x_{j_{0}}^{0}}\right\}
_{j_{0},j_{L}=1}^{n_{0},n_{L}}=\prod_{l=1}^{L}n^{-1/2}_{l-1} D^{l}X^{l},\;\mathbf{n}%
_{L}=(n_{1},...,n_{L})  \label{jac}
\end{equation}%
i.e., a $n_{L}\times n_{0}$ random matrix, where
\begin{equation}
D^{l}=\{D_{j_{l}}^{l}\delta
_{j_{l}k_{l}}\}_{j_{l},k_{l}=1}^{n_{l}},\;D_{j_{l}}^{l}=\varphi ^{\prime }%
\Big(n_{l-1}^{-1/2}%
\sum_{j_{l-1}=1}^{n_{l-1}}X_{j_{l}j_{l-1}}^{l}x_{j_{l-1}}^{l-1}+b_{j_{l}}^{l}%
\Big)  \label{D}
\end{equation}%
are diagonal random matrices.

We are interested in the spectrum of singular values of $J_{\mathbf{n}%
_{L}}^{L}$, i.e., the square roots of eigenvalues of
\begin{equation}
M_{\mathbf{n}_{L}}^{L}:=J_{\mathbf{n}_{L}}^{L}(J_{\mathbf{n}_{L}}^{L})^{T}
\label{JJM}
\end{equation}%
for networks with the above random weights and biases and for large $%
\{n_{l}\}_{l=1}^{L}$, i.e., for deep networks with wide layers, see \cite%
{Li-Qi:00,Ma-Co:16,Pe-Ba:17,Pe-Co:18,Po-Co:16,Sc-Co:17} for motivations and
settings. More precisely, we will study in this paper the asymptotic case
determined by the simultaneous limits
\begin{equation}
\lim_{N_{l}\rightarrow \infty }\frac{n_{l-1}}{n_{l}}=c_{l}\in (0,\infty
),\;n_{l}\rightarrow \infty ,\;l=1,...,L  \label{asf1}
\end{equation}%
denoted below as
\begin{equation}
\lim_{\mathbf{n}_{L}\rightarrow \infty }.  \label{limn}
\end{equation}%
%
Denote $\{\lambda _{t}^{L}\}_{t=1}^{n_{L}}$ the eigenvalues of the $n_L
\times n_L$ random matrix $M_{\mathbf{n}_{L}}^{L}$ and define its Normalized
Counting Measure (NCM) as
\begin{equation}
\nu _{M_{\mathbf{n}_{L}}^{L}}:=n_{L}^{-1}\sum_{t=1}^{N_{L}}\delta _{\lambda
_{t}^{L}}.  \label{ncm}
\end{equation}%
We will deal with the leading term of $\nu _{M_{\mathbf{n}_{L}}^{L}}$in the
asymptotic regime (\ref{asf1}) -- (\ref{limn}), i.e., with the limit%
\begin{equation}
\nu _{M^{L}}:=\lim_{\mathbf{n}_{L}\rightarrow \infty }\nu _{M_{\mathbf{n}%
_{L}}^{L}}  \label{ids}
\end{equation}%
if any. Note that since $\nu _{M_{\mathbf{n}_{L}}^{L}}$ is random, the
meaning of the limit has to be stipulated.

The problem was considered in \cite{Pe-Co:18} (see also \cite%
{Li-Qi:00,Pe-Ba:17}) in the case where all $b^{l}$ and $X^{l},\;l=1,2,...,L$
in (\ref{bga}) -- (\ref{wga}) are Gaussian and have the same size $n$ and $%
n\times n$ respectively, i.e.,%
\begin{equation}  \label{nequ}
n:=n_{0}=...=n_{L}.
\end{equation}%
We will write in this case $n$ instead of $n_l, \; l=0,...,L$. In \cite{Pe-Co:18}
compact formulas for the limit
\begin{equation}
\overline{\nu }_{M^{L}}:=\lim_{n\rightarrow \infty }\overline{\nu }%
_{M_{n}^{L}},\;\overline{\nu }_{M_{n}^{L}}:=\mathbf{E}\{\nu _{M_{n}^{L}}\}
\label{mids}
\end{equation}
and its Stieltjes transform
\begin{equation}
f_{M^{L}}(z)=\int_{\infty }^{\infty }\frac{\overline{\nu }_{M^{L}}(d\lambda )%
}{\lambda -z},\;\Im z\neq 0  \label{stm}
\end{equation}%
were proposed. The formula for $\overline{\nu} _{M^{L}}$ is given in (\ref%
{nucon}) below. To write the formula for $f_{M^{L}}$ it is convenient to use
the moment generating function
\begin{equation}
m_{M^{L}}(z)=\sum_{k=1}^{\infty }m_{k}z^{k},\;m_{k}=\int_{\infty }^{\infty
}\lambda ^{k}\overline{\nu }_{M^{L}}(d\lambda ),  \label{mgen}
\end{equation}
related to $f_{M^{L}}$ as
\begin{equation}
m_{M^{L}}(z)=-1-z^{-1}f_{M^{L}}(z^{-1}).  \label{stmg}
\end{equation}%
Let
\begin{equation}
K_{n}^{l}:=(D_{n}^{l})^{2}=\{(D_{j_{l}}^{l})^{2}\}_{j_{l}=1}^{n}  \label{kan}
\end{equation}%
be the square of the $n\times n$ random diagonal matrix (\ref{D}) with $n_{l}=n$,
denoted $D_{n}^{l}$ to make explicit its dependence on $n$ of (\ref{nequ}), and let $m_{K^{l}}$ be the moment generating function of the $n\rightarrow
\infty $ limit $\overline{\nu }_{K^{l}}$ of the expectation of the NCM of $%
K_{n}^{l}$. Then we have according to formulas (14) and (16) in \cite%
{Pe-Co:18} in the case where $\overline{\nu }_{K^{l}}$, hence $m_{K^{l}}$, do not depend on $l$
(see Remark \ref{r:penn} (i))
\begin{eqnarray}
m_{M^{L}}(z) &=&m_{K}(z^{1/L}\Psi _{L}(m_{M^{L}}(z))),  \label{penfo} \\
\Psi _{L}(z) &=&(1+z)^{1/L}z^{1-1/L}.  \notag
\end{eqnarray}%
i.e., $f_{M^{L}}$ of (\ref{stm})) satisfies a certain functional equation,
the standard situation in random matrix theory and its applications, see
\cite{Pa-Sh:11} for general results and \cite{Go-Co:15,Mu:02} for results on
the products of random matrices. Note that our notation is different from
that of \cite{Pe-Co:18}: our $f_{M^{L}}(z)$ of (\ref{stm}) is $-G_{X}(z)$ of
(7) in \cite{Pe-Co:18} and our $m_{M^{L}}(z)$ of (\ref{mgen}) is $M_{X}(1/z)$
of (9) in \cite{Pe-Co:18}.

The derivation of (\ref{penfo}) and the corresponding formula for the
limiting mean NCM\ $\overline{\nu }_{M^{L}}$ in \cite{Pe-Co:18} was based on
the claimed there asymptotic freeness of diagonal matrices $%
D_{n_l}^{l}=\{D_{j_{l}}^{l}\}_{j_{l}=1}^{n_l},\;l=1,2...,L$ of (\ref{D})
and Gaussian matrices $X_{n_l}^{l},\;l=1,2,...,L$ of (\ref{wl}) -- (%
\ref{wga})
(see, e.g. \cite%
{Ch-Co:18,Mi-Sp:17,Pe-Hi:00} for the definitions and properties of
asymptotic freeness). This leads directly to (\ref{penfo}) in view of the
multiplicative property of the moment generating functions (\ref{mids})
and the so-called $S$-transforms of $\overline{\nu }_{K^{l}}$
and of $\nu_{MP}$, the mean limiting NCM's of $K_{n_l}^{l}$ and of $n^{-1}X_{n_l}^l (X_{n_l}^{l})^T$ in the regime (\ref{asf1}), see Remark \ref{r:penn} (ii) and Corollary \ref{c:conv}.

There is, however, a delicate point in the proof of the above results in
\cite{Pe-Co:18}, since, to the best of our knowledge, the asymptotic
freeness has been established so far for the Gaussian random matrices $X^{l}_{n_l}$ of (\ref{wga}) and deterministic (more generally, random but $%
X^{l}_{n_l}$-independent) diagonal matrices, see e.g. \cite{Ch-Co:18,Mi-Sp:17,Pe-Hi:00} and also
\cite{Go-Co:15,Mu:02} treating the product matrices of form (\ref{JJM}) with  $X^l_n$-independent diagonal matrices.
On the other
hand, the diagonal matrices $D_{n}^{l}$ in (\ref{D}) depend explicitly on $%
(X_{n}^{l},b_{n}^{l})$ of (\ref{wl}) -- (\ref{bl}) and, implicitly, via $%
x^{l-1}$, on the all "preceding" $(X_{n}^{l^{\prime }},b_{n}^{l^{\prime
}}),\;l^{\prime }=l-1,...,1$. Thus, the proof of validity of (\ref{penfo})
requires an additional reasoning. The goal of this paper is to provide this
reasoning, thereby justifying the basic formula (\ref{penfo}) and the
corresponding formulas for the mean limiting NCM $\overline{\nu }_{M^{L}}$
of (\ref{mids}), see formula (13) of \cite{Pe-Co:18} and formula (\ref%
{nucone}) below. Moreover, we prove that the formula (\ref{ids}) is valid
not only in the mean (see (\ref{mids})), but also with probability 1 in $%
\Omega _{L}$ of (\ref{oml}) (recall that the measures in the r.h.s. of (\ref%
{ids})\ are random) and that the limiting measure $\nu _{M^{L}}$ in the
l.h.s. of (\ref{ids}) coincides with $\overline{\nu }_{M^{L}}$ of (\ref{mids}%
), i.e., $\nu _{M^{L}}$ is nonrandom.

Note that a possible version of the proof of the above assertions could be
carried out by extending the corresponding proofs in free probability (see, e.g. \cite{Mi-Sp:17,Pe-Hi:00}) to the case
where the diagonal matrices are given by (\ref{D}). We will prefer, however,
another approach based on the standard techniques of random matrix theory, see e.g. \cite{Pa-Sh:11}. There the main technical tools are some differentiation
formulas (see, e.g. (\ref{difg})), providing certain identities for
expectations of essential spectral characteristics, and bounds (Poincar\'e,
martingale) for the variance of these characterizes, guaranteing the
vanishing of their fluctuations in the large size (layer width) limit,
thereby allowing for the conversion of the obtained identities into
functional equations for the characteristics in question, the Stieltjes
transform of the limiting NCM in particular. This, however, has to be
complemented (in fact, preceded) by a certain assertion (see Lemma \ref%
{l:inter}) justifying the asymptotic replacement of random $X_{n_l}^l$-dependent
matrices $D_{n_l}^l$ in (\ref{jac}) -- (\ref{D}) by certain random but $X_{n_l}^l$%
-independent matrices (see (\ref{mnbf}) -- (\ref{kab1})) and allowing us not
only to substantiate the results of \cite{Pe-Co:18}, but also to extend them
to the case of i.i.d. but not necessarily Gaussian $(X_{n_l}^{l},b_{n_l}^{l}),%
\;l=1,...,L$ \cite{Pa-Sl:20}.

The paper is organized as follows. In the next section we prove the validity
of (\ref{ids}) with probability 1 in $\Omega_L$ of (\ref{oml}), formula (\ref%
{penfo}) and the corresponding formula for $\nu _{M^{L}}=\overline{\nu }%
_{M^{L}}$ of \cite{Pe-Co:18}. The proofs are based on a natural inductive
procedure allowing for the passage from the $l$th to the $(l+1)$th layer. In
turn, the induction procedure is based on a formula relating the limiting
(in the layer width) Stieltjes transforms of the NCM's of two subsequent
layers. The formula is more or less standard both in its form and its
derivation in the case where the matrices $D_{n}^{l}$ in (\ref{jac}) are
deterministic or random but independent of $(X_{n}^{l^{\prime
}},b_{n}^{l^{\prime }}),\;l^{\prime }=l,l-1,...,1$, see e.g. \cite{Co-Ha:14,Pa-Sh:11} and references therein.%
The case of dependent $D_{n}^{l}$ as in (\ref{D}) is treated in Section 3.

It follows from the results of the section that the coincidence of the limiting eigenvalue distribution of matrices of two indicated cases is due to the form of
dependence of $D_{n}^{l}$ on $(X_{n}^{l^{\prime }},b_{n}^{l^{\prime
}}),\;l^{\prime }=l,l-1,...,1$ given by (\ref{D}), which is, so to say,
"slow varying" and does not contribute to the leading term (the limit (\ref%
{ids})) of the corresponding eigenvalue distribution.

\section{Main Result and its Proof.}
As was already mention
ed in Introduction, our goal is to present a more complete
proof of the results of work \cite{Pe-Co:18} by using random matrix theory.
Thus, to formulate our results, we need several facts of the theory.

Consider for every positive integer $n$: (i) the $n\times n$ random matrix $%
X_{n}$ with entries which are independent standard (mean zero and variance
1) Gaussian random variables; (ii) positive definite matrices $\mathsf{K}_{n}$ and $\mathsf{R}_{n}$ that may also be random but independent of $X_n$ and such that their Normalized Counting
Measures $\nu _{\mathsf{K}_{n}}$ and  $\nu _{\mathsf{R}_{n}}$ (see (\ref{ncm}))
converge weakly (with probability 1 if random) as $n\rightarrow
\infty $ to non-random measures $\nu _{\mathsf{K}}$ and
$\nu _{\mathsf{R}}$. Set $%
\mathsf{M}_{n}=n^{-1}\mathsf{R}_{n}^{1/2}X^{T}\mathsf{K}_{n}X_{n}
\mathsf{R}_{n}^{1/2}$. According to random matrix theory (see, e.g. Lemma \ref{l:rkx} below, \cite{Co-Ha:14} and references therein), in this case and under some conditions the  Normalized Counting Measure $\nu _{\mathsf{M}_{n}}$ of $\mathsf{M}_{n}$
converges weakly with probability 1 as $n\rightarrow \infty $ to a non-random
measure $\nu _{\mathsf{M}}$ which is uniquely determined by the limiting measures
 $\nu _{\mathsf{K}}$ and $\nu _{\mathsf{R}}$
via a certain analytical procedure (see, e.g.
formulas (\ref{stm}) and (\ref{fh}) -- (\ref{kh}) below). We can write down this fact as%
\begin{equation}
\nu _{\mathsf{M}}=\nu _{\mathsf{K}}\diamond \nu _{\mathsf{R}} \label{opdia}
\end{equation}
and say that the procedure defines a binary operation in the set
of non-negative measures with the total mass 1 and a support belonging to the
positive semiaxis
 (see more details in Remark \ref{r:penn} (ii), Lemma \ref{l:rkx} and Corollary \ref{c:conv}).

This is the result of \cite{Pe-Co:18} confirmed in the present paper that the
limiting Normalized Counting Measure (\ref{ids}) of random matrices
(\ref{JJM}), where the role of $\mathsf{K}_{n}$ plays the matrix given by (\ref{kan}) and (\ref{D}) and depending on the Gaussian matrices $X^{l}$'s of (\ref{wga}), can be found as the "product" with respect the operation (\ref{opdia}) of $L$ measures $\nu_{K^l}, \; l=1,...,L$ which are indicated in Theorem \ref{t:main} and are the limiting Normalized Counting Measures of special random matrices that do not depend on $X^l$'s of (\ref{wga}), see (\ref{kab1}) and the subsequent text.

Note that (\ref{opdia}) is a version of the so-called multiplicative
convolution of free probability theory \cite{Mi-Sp:17,Pe-Hi:00}, having the
above random matrices as a basic analytic model.

We will follow \cite{Pe-Co:18} and confine
ourselves to the case (\ref{nequ}) where all the weight matrices and bias
vectors are of the same size $n$. The general case of
different sizes is essentially the same (see, e.g. Remark \ref{r:rra}
(iii)). 

\begin{theorem}
\label{t:main} Let $M_{n}^{L}$ be the random matrix (\ref{JJM}) defined by (%
\ref{rec}) -- (\ref{D}) and (\ref{nequ}), where the biases $b^{l}$ and
weights $W^{l}$ are random Gaussian variables satisfying (\ref{bga}) -- (\ref{wga})
and the input vector $x^{0}$ (\ref{x0}) (deterministic or random) is such
that there exists a finite limit%
\begin{equation}
q^{1}:=\lim_{n\rightarrow \infty }q_{n}^{1}>\sigma
_{b}^{2}>0,\;\;q_{n}^{1}=n^{-1}\sum_{j_{0}=1}^{n}(x_{j_{0}}^{0})^{2}+\sigma
_{b}^{2}.  \label{q0}
\end{equation}%
Assume also that the nonlinearity $\varphi $ in (\ref{rec}) is a piecewise
differentiable function such that $\varphi ^{\prime }$ is not zero
identically and denote%
\begin{equation}
\sup_{t\in \mathbb{R}}|\varphi (t)|=\Phi _{0}<\infty ,\;\sup_{t\in \mathbb{R}%
}|\varphi ^{\prime }(t)|=\Phi _{1}<\infty .\;\;  \label{phi1}
\end{equation}%
Then the Normalized Counting Measure (NCM) $\nu _{M_{n}^{L}}$ of $M_{n}^{L}$
(see (\ref{ncm})) converges weakly with probability 1 in the probability
space $\Omega _{L}$ of (\ref{oml}) to the non-random measure
\begin{equation}
\nu _{M^{L}}=\nu _{K^{1}}\diamond \nu _{K^{2}}...\diamond \nu
_{K^{L}}\diamond \delta _{1},  \label{nucon}
\end{equation}%
where the operation "$\diamond $" is defined in (\ref{opdia}) (see also Lemma \ref{l:rkx} and Corollary \ref{c:conv}),
$\delta _{1}$ is the unit measure
concentrated at 1 and $\nu _{K^{l}},\;l=1,...,L$ is the probability
distribution of the random variable \ $(\varphi ^{\prime }(\gamma \sqrt{q^{l}%
}))^{2}$ with the standard Gaussian random variable $\gamma $ and $q^{l}$
determined by the recurrence
\begin{equation}
q^{l}=\mathbf{E}\{\varphi ^{2}(\gamma \sqrt{q^{l-1}})\},\;l\geq 2,
\label{qlql}
\end{equation}%
with $q^{1}$ given by (\ref{q0}).
\end{theorem}

\begin{remark}
\label{r:penn} (i) If
\begin{equation}
q_{L}= \cdots =q_{1},  \label{qeq}
\end{equation}%
then $\nu _{K}:=\nu _{K^{l}},$ $l=1,...,L$, (\ref{nucon}) becomes
\begin{equation}
\nu _{M^{L}}=\underset{L\;\mathrm{times}}{\underbrace{\nu _{K}\diamond \nu
_{K} \cdots \diamond \nu _{K}}\diamond \delta _{1}}. \label{nucone}
\end{equation}%
An important case of equalities (\ref{qeq}) is where $q^{1}=q^{\ast }$
and $q^{\ast }$
is the fixed point of (\ref{qlql}), see \cite{Ma-Co:16,Po-Co:16,Sc-Co:17}
for a detailed analysis of (\ref{qlql}) and its role in the deep neural
networks functioning.

(ii) Let us show now that Theorem \ref{t:main} implies the results of \cite%
{Pe-Co:18}. It follows from the theorem, (\ref{mlml1}), and Corollary \ref{c:conv} that the functional inverse $%
z_{M^{l+1}}$ of the moment generating function $m_{M^{l+1}}$ (see (\ref{mgen}%
) -- (\ref{stmg})) of the limiting NCM $\nu _{M^{l+1}}$ of matrix $%
M_{n}^{l+1}$ and that of $M_{n}^{l}$ are related as in (\ref{mconv}), i.e.,
\begin{equation}
z_{M^{l+1}}(m)=z_{K^{l+1}}(m)z_{M^{l}}(m)m^{-1}.  \label{zre}
\end{equation}%
Passing from the moment generating functions to the S-transforms of free probability theory via the
formula $S(m)=(1+m)m^{-1}z(m)$  and taking into account that the S-transform  of the limiting NCM of the
Wishart matrix $n^{-1}X_{n}X_{n}^{T}$ is $S_{MP}=(1+m)^{-1}$  (see \cite{Mi-Sp:17}), we
obtain from (\ref{zre})%
\begin{equation}
S_{M^{l+1}}(m)=S_{K^{l+1}}(m)S_{MP}(m)S_{M^{l}}(m).  \label{prec}
\end{equation}%
Iterating this relation from $l=1$ to $l=L-1$, we obtain formula (13) of
\cite{Pe-Co:18}.
The functional equation (\ref{penfo}) arising in the case (\ref{qeq}) of
the $l$-independent parameters $q_{l}$ of (\ref{qlql}) is derived from (\ref%
{prec}) in \cite{Pe-Co:18}.

Recall now that according to free probability
\cite{Ch-Co:18,Mi-Sp:17} if $\mathsf{K_n}$ and $\mathsf{R_n}$ are
positive definite matrices,
$\mathsf{N_n}=\mathsf{K_n}^{1/2}\mathsf{R_n}\mathsf{K_n}^{1/2}$
and all of them have limiting Normalized Counting Measures
$\nu_{\mathsf{K}}$, $\nu_{\mathsf{R}}$ and $\nu_{\mathsf{N}}$,
then their S-transforms are related as
$S_{\mathsf{N}}$=$S_{\mathsf{K}}S_{\mathsf{R}}$. One says in this
case that the limiting Normalized Counting Measure of
$\mathsf{N}_n$ is the free convolution of those of $\mathsf{K}_n$
and $\mathsf{R}_n$ and writes
$\nu_{\mathsf{N}}=\nu_{\mathsf{K}}\boxplus \nu_{\mathsf{R}}$.
Comparing this with (\ref{opdia}) and (\ref{prec}), we conclude
that the operations $\diamond$ and $\boxplus$ are related as
$\nu_{\mathsf{K}}\diamond \nu_{\mathsf{R}}$ =$\nu_{\mathsf{K}}
\boxplus \nu_{MP}\boxplus\nu_{\mathsf{R}}$.

(iii) In the subsequent work \cite{Pa-Sl:20} we consider a more general case
of not necessarily Gaussian random variables, i.e., where the entries of
independent random matrices $X^{l},\;l=1,2,...$ in (\ref{jac}) -- (\ref{D})
are i.i.d. random variables satisfying (\ref{wga}) and certain moment
conditions and the component of independent vectors $b^{l},\;l=1,2,...$ are
i.i.d. random variables satisfying (\ref{bga}). It is shown that in this,
more general case, the conclusion of the theorem is still valid, however the
measure $\nu _{K^{l}},\;l=1,2,...$ is now the probability distribution of $%
(\varphi ^{\prime }(\gamma \sqrt{(q^{l-1}-\sigma _{b}^{2})}+b_{1}^{l}))^{2}$%
, where $\gamma $ is again the standard Gaussian random variable and (\ref%
{qlql}) is replaced by
\begin{equation}
q^{l}=\int \varphi ^{2}\Big(\gamma \sqrt{q^{l-1}-\sigma _{b}^{2}}+b\Big)%
\Gamma (d\gamma )F(db),\;l\geq 2,  \label{qlqlg}
\end{equation}%
where $\Gamma (d\gamma )=(2\pi )^{1/2}e^{-\gamma ^{2}/2}d\gamma $, $F$ is
the probability law of $b_{1}^{l}$ and $q^{1}$ is again given by (\ref{q0}).

(iv) If the input vector (\ref{x0}) are random, then it is assumed that they
are defined on the same probability space $\Omega _{x^{0}}$ for all $n_{0}$
and the limit $q^{1}$ exists with probability 1 in $\Omega _{x^{0}}$. An
example of this situation is where $\{x_{j_{0}}^{l}\}_{j^{0}=1}^{n_{0}}$ are
the first $n_{0}$ components of an ergodic sequence $\{x_{j_{0}}^{l}%
\}_{j^{0}=1}^{\infty }$ (e.g. a sequence of i.i.d. random variables) with
finite second moment. Here $q_{1}$ in (\ref{q0}) exists with probability 1
on $\Omega _{x^{0}}$ and even is non-random just by ergodic theorem (the
strong Law of Large Numbers in the case of i.i.d sequence) and the theorem
is valid with probability 1 in $\Omega _{l}\times \Omega _{x^{0}}$.
\end{remark}
We present now the proof of Theorem \ref{t:main}.

\smallskip
\begin{proof}
We prove the
theorem by induction in $L$. We have from (\ref{rec}) -- (\ref{JJM}) and (%
\ref{nequ}) with $L=1$ the following $n\times n$ matrix%
\begin{equation}
M_{n}^{1}=J_{n}^{1}(J_{n}^{1})^{T}=n^{-1}D_{n}^{1}X_{n}^{1}(X_{n}^{1})^{T}D_{n}^{1}.
\label{m1}
\end{equation}%
It is convenient to pass from $M_{n}^{1}$ to the $n\times n$ matrix%
\begin{equation}
\mathcal{M}%
_{n}^{1}=(J_{n}^{1})^{T}J_{n}^{1}=n^{-1}(X_{n}^{1})^{T}K_{n}^{1}X_{n}^{1},%
\;K_{n}^{1}=(D_{n}^{1})^{2}  \label{cm1}
\end{equation}%
which has the same spectrum, hence the same Normalized Counting Measure as $%
M_{n}^{1}$. The matrix $\mathcal{M}_{n}^{1}$ is a particular case with $%
S_{n}=\mathbf{1}_{n}$ of matrix (\ref{mncal}) treated in Theorem \ref{t:ind} below.
Since the NCM of the unit matrix $1_{n}$ is the Dirac measure $\delta _{1}$,
conditions (\ref{r2}) -- (\ref{nur}) of the theorem are evident. Condition (\ref{qqn}) of the
theorem is just (\ref{q0}). It follows then from
Corollary \ref{c:conv}  that the assertion of our theorem, i.e.,
formula (\ref{nucon}) with $q^{1}$ of (\ref{q0}) is valid for $L=1$.

Consider now the case $L=2$ of (\ref{rec}) -- (\ref{JJM}) and (\ref{nequ}):
\begin{equation}
M_{n}^{2}=n^{-1}D_{n}^{2}X_{n}^{2}M_{n}^{1}(X_{n}^{2})^{T}D_{n}^{2}.
\label{m2m1}
\end{equation}%
Since $M_{n}^{1}$ is positive definite, we have
\begin{equation}
M_{n}^{1}=(S_{n}^{1})^{2}  \label{m1s1}
\end{equation}%
with a positive definite $S_{n}^{1}$, hence
\begin{equation}
M_{n}^{2}=n^{-1}D_{n}^{2}X_{n}^{2}(S_{n}^{1})^{2}(X_{n}^{2})^{T}D_{n}^{2}
\label{mn2}
\end{equation}%
and the corresponding $\mathcal{M}_{n}^{2}$ is
\begin{equation}
\mathcal{M}%
_{n}^{2}=n^{-1}S_{n}^{1}(X_{n}^{^{2}})^{T}K_{n}^{2}X_{n}^{2}S_{n}^{1},%
\;K_{n}^{2}=(D_{n}^{2})^{2}.  \label{cm1m2}
\end{equation}%
We observe that $\mathcal{M}_{n}^{2}$ is a particular case of matrix (\ref%
{mncal}) of Theorem \ref{t:ind} with $M_{n}^{1}=(S_{n}^{1})^{2}$ as $%
R_{n}=(S_{n})^{2}$, $X_{n}^{2}$ as $X_{n}$, $K_{n}^{2}$ as $K_{n}$, $%
\{x_{j_{1}}^{1}\}_{j_{1}=1}^{n}$ as $\{x_{\alpha n}\}_{\alpha =1}^{n}$, $%
\Omega _{1}=\Omega ^{1}$ of (\ref{oml}) as $\Omega _{Rx}$ and $\Omega ^{2}$
of (\ref{oml}) as $\Omega _{Xb}$, i.e., the case of the random but\ $%
\{X_{n}^{2},b^{2}_{n}\}$ -independent $R_{n}$ and $\{x_{\alpha n}\}_{\alpha
=1}^{n}$ in (\ref{mncal}) as described in Remark \ref%
{r:rra} (i). Let us check that conditions (\ref{r2}) -- (\ref{nur}) \ and (\ref%
{qqn}) of Theorem \ref{t:ind} are satisfied for $\mathcal{M}_{n}^{2}$ of (%
\ref{cm1m2}) with probability 1 in the probability space $\Omega
_{1}=\Omega^1$ generated by $\{X_{n}^{1},b_{n}^{1}\}$ for all $n$ and
independent of the space $\Omega ^{2}$ generated by $\{X_{n}^{2},b_{n}^{2}\}$
for all $n$.

We will need here an important fact on the operator norm of $n\times n$
random matrices with independent standard Gaussian entries. Namely, if $%
X_{n} $ is such $n\times n$ matrix, then we have with probability 1
\begin{equation}
\lim_{n\rightarrow \infty }n^{-1/2}||X_{n}||=2,  \label{lin}
\end{equation}%
thus, with the same probability%
\begin{equation}
||X_{n}||\leq Cn^{1/2},\;C > 2  \label{nox}
\end{equation}%
if $n$ is large enough.

For the Gaussian matrices relation (\ref{lin}) has already been known in the
Wigner's school of the early 1960th, see \cite{Pa-Sh:11}. It follows in this
case from the orthogonal polynomial representation of the density of the
 NCM of $n^{-1}X_nX_n^T$ and the asymptotic formula for the
corresponding orthogonal polynomials. For the modern form of
(\ref{lin}) and (\ref{nox}), in particular their validity for
i.i.d matrix entries with mean zero and finite fourth moment, see
\cite{Ba-Si:10,Ve:18} and references therein.

We will also need the bound%
\begin{equation}
||K_{n}^{1}||\leq (\Phi _{1})^{2},  \label{nok}
\end{equation}%
following from (\ref{D}), (\ref{kan}) and (\ref{phi1}) and valid everywhere in $\Omega _{1}$ of (\ref{oml}).

Now, by using (\ref{cm1}), (\ref{nox}), (\ref{nok}) and the inequality
\begin{equation}
|\mathrm{Tr}AB|\leq ||A||\mathrm{Tr}B,\;  \label{tab}
\end{equation}%
valid for any matrix $A$ and a positive definite matrix $B$, we obtain with
probability 1 in $\Omega _{1}$ and for sufficiently large $n$
\begin{equation*}
n^{-1}\mathrm{Tr}(M_{n}^{1})^{2}=n^{-3}\mathrm{Tr}\mathbb{(}%
K_{n}^{1}X_{n}^{1}(X_{n}^{1})^{T})^{2}\leq (C\Phi _{1})^{4}.
\end{equation*}%
We conclude that $M_{n}^{1}$, which plays here the
role of $R_{n}$ of Theorem \ref{t:ind} and Remark \ref{r:rra} (i) according to  (\ref{m1s1}), satisfies condition (\ref{r2}) with $%
r_{2}=(C\Phi _{1})^{4}$ and with probability 1 in our case, i.e., on a
certain $\Omega _{11}\subset \Omega _{1},\;\mathbf{P}(\Omega _{11})=1$

Next, it follows from the above proof of the theorem for $L=1$, i.e., in
fact, from Theorem \ref{t:ind}, that there exists $\Omega _{12}\subset
\Omega _{1},\;\mathbf{P}(\Omega _{12})=1$ on which the NCM $\nu _{M_{n}^{1}}$
converges weakly to a non-random limit $\nu _{M^{1}}$, hence condition (\ref%
{nur}) is also satisfied with probability 1, i.e., on $\Omega _{12}$.

At last, according to Lemma \ref{l:xlyl} (i) and (\ref{q0}), there exists $%
\Omega _{13}\subset \Omega _{1},\;\mathbf{P}(\Omega _{13})=1$ on which there
exists
\begin{equation*}
\lim_{n\rightarrow \infty
}n^{-1}\sum_{j_{1}=1}^{n}(x_{j_{1}}^{1})^{2}+\sigma _{b}^{2}=q^{2}>\sigma
_{b}^{2},
\end{equation*}%
i.e., condition (\ref{qqn}) is also satisfied.

Hence, we can apply Theorem \ref{t:ind}\ on the subspace $\overline{\Omega }%
_{1}=\Omega _{11}\cap \Omega _{12}\cap \Omega _{13}\subset \Omega _{1},\;%
\mathbf{P}(\overline{\Omega }_{1})=1$ where all the conditions of the theorem %
are valid, i.e., $\overline{\Omega }_{1}$ plays the role of $%
\Omega _{Rx} $ of Remark \ref{r:rra} (i). Thus the theorem implies that for
any $\omega _{1}\in \overline{\Omega_1 }$ there exists subspace $\overline{%
\Omega^2}(\omega _{1})$ of the space $\Omega ^{2}$ generated by $%
\{X_{n}^{2},B_{n}^{2}\}$ for all $n$ and such that
$\mathbf{P}(\overline{\Omega^2}(\omega _{1}))=1$ and formulas (\ref{nucon}) -- (\ref{qlql}) are valid
for $L=2$. It follows then from the Fubini theorem that the same is true on
a certain $\overline{\Omega }_{2}\subset \Omega _{2},\;\mathbf{P}(\overline{%
\Omega }_{2})=1$ where $\Omega _{2}$ is defined by (\ref{oml}) with $L=2$.

This proves the theorem for $L=2$. The proof for $L=3,4,...$ is analogous,
since (cf. (\ref{mn2}))
\begin{equation}
M_{n}^{l+1}=n^{-1}D_{n}^{l+1}X_{n}^{l+1}M_{n}^{l}(X_{n}^{l+1})^{T}D_{n}^{l+1},\;l\geq 2.
\label{mlml1}
\end{equation}%
In particular, we have with probability 1 on $\Omega _{l-1}$ of (\ref{oml})
for $M_{n}^{l-1}$ playing the role of $R_{n}$ of Theorem \ref{t:ind} on the $%
l$th step of the inductive procedure (cf. (\ref{r2}))
\begin{equation*}
n^{-1}\mathrm{Tr}(M^{l})^{2}\leq (C\Phi _{1})^{4l},\;l\geq 2.
\end{equation*}%
\end{proof}

\section{Auxiliary Results.}

Our main result, Theorems \ref{t:main} on the limiting eigenvalue
distribution of random matrices (\ref{JJM}) for any $L$, is proved above by
induction in the layer number $l$, see formulas (\ref{m2m1}), (\ref{cm1m2})
and (\ref{mlml1}). To carry out the passage from the $l$th to the $(l+1)$th
layer we need an expression for the limiting NCM $\nu _{\mathcal{M}^{l+1}}$
of the matrix $\mathcal{M}_{n}^{l+1}$ via that of $\mathcal{M}_{n}^{l}$ in
the infinite width limit $n\rightarrow \infty $. The corresponding results,
which could be of independent interest, as well as certain auxiliary results
are proved in this section. In particular, a functional equation relating
the Stieltjes transform of $\nu _{\mathcal{M}_{n}^{l+1}}$ and $\nu _{%
\mathcal{M}_{n}^{l}}$ in the limit $n\rightarrow \infty $ is obtained.

\begin{theorem}
\label{t:ind} Consider for every positive integer $n$ the $n\times n$ random matrix
\begin{equation}
\mathcal{M}_{n}=n^{-1}S_{n}X_{n}^{T}K_{n}X_{n}S_{n},  \label{mncal}
\end{equation}%
where:

\smallskip
(a) $S_{n}$  is a positive definite $n\times n$ matrix
such that
\begin{equation}
\sup_{n}n^{-1}\mathrm{Tr}R_{n}^{2}=r_{2}<\infty ,\;R_{n}=S_{n}^{2},
\label{r2}
\end{equation}%
and
\begin{equation}
\lim_{n\rightarrow \infty }\nu _{R_{n}}=\nu _{R},\;\nu _{R}(\mathbb{R}%
_{+})=1,  \label{nur}
\end{equation}%
where $\nu _{R_{n}}$ is the Normalized Counting Measure of $R_{n}$, $\nu
_{R} $ is a non-negative measure not concentrated at zero and $%
\lim_{n\rightarrow \infty }$ denotes here the weak convergence of
probability measures;

\smallskip
(b) $X_{n}$ is the $n\times n$ random matrix
\begin{equation}
X_{n}=\{X_{j\alpha }\}_{j,\alpha =1}^{n},\;\mathbf{E}\{X_{j\alpha }\}=0,\;%
\mathbf{E}\{X_{j_{1}\alpha _{1}}X_{j_{2}\alpha _{2}}\}=\delta
_{j_{1}j_{2}}\delta _{\alpha _{1}\alpha _{2}},  \label{Xn}
\end{equation}%
with the independent standard Gaussian entries (cf.
(\ref{wga})), $b_n$ is the $n$-component random vector
\begin{equation}
b_{n}=\{b_{j}\}_{j=1}^{n},\;\mathbf{E}\{b_{j}\}=0,\;\mathbf{E}%
\{b_{j_{1}}b_{j_{2}}\}=\sigma _{b}^{2}\delta _{j_{1}j_{2}}  \label{b}
\end{equation}%
with the independent Gaussian components of zero mean and variance $\sigma _{b}^{2}$
(cf. (\ref{bga})) and for all $n$ matrix $X_{n}$ and the vector $b_{n}$
viewed as defined on the probability space%
\begin{equation}
\Omega _{Xb}=\Omega _{X}\times \Omega _{b},  \label{oxb}
\end{equation}
where $\Omega _{X}$ and $\Omega _{b}$ are generated by (\ref{xinf}) and (\ref%
{binf});

\smallskip
(c) $K_n$ and $D_n$ are the diagonal random matrices
\begin{equation}
K_{n}=D_{n}^{2},\;D_{n}=\{\delta _{jk}D_{jn}\}_{j,k=1}^{n},\;D_{jn}=\varphi
^{\prime }\Big(n^{-1/2}\sum_{a=1}^{n}X_{j\alpha }x_{\alpha n}+b_{j}\Big),
\label{Dn}
\end{equation}%
where $\varphi :\mathbb{R}\rightarrow \mathbb{R}$ is a piecewise
differentiable function, such that (cf. (\ref{phi1}))
\begin{equation}
\sup_{x\in \mathbb{R}}|\varphi (x)|=\Phi _{0}<\infty ,\;\sup_{x\in \mathbb{R}%
}|\varphi ^{\prime }(x)|=\Phi _{1}<\infty ,  \label{Phi}
\end{equation}%
and $x_{n}=\{x_{\alpha n}\}_{\alpha =1}^{n}$ is a collection of real numbers
such that there exists%
\begin{equation}
q=\lim_{n\rightarrow \infty }q_{n}>\sigma
_{b}^{2}>0,\;q_{n}=n^{-1}\sum_{\alpha =1}^{n}(x_{\alpha n})^{2}+\sigma
_{b}^{2}.  \label{qqn}
\end{equation}%
Then the Normalized Counting Measure (NCM) $\nu _{\mathcal{M}_{n}}$ of $%
\mathcal{M}_{n}$ converges weakly with probability 1 in $\Omega _{Xb}$ of (%
\ref{oxb}) to a non-random measure $\nu _{\mathcal{M}}$ whose Stieltjes
transform $f_{\mathcal{M}}$ (see (\ref{stm})) can be obtained from the
formulas%
\begin{equation}
f_{\mathcal{M}}(z)=\int_{0}^{\infty }\frac{\nu _{R}(d\lambda )}{k(z)\lambda
-z}=-z^{-1}+z^{-1}h(z)k(z),  \label{fh}
\end{equation}%
where the pair ($h,k$) is a unique solution of the system of functional
equations%
\begin{equation}
h(z)=\int_{0}^{\infty }\frac{\lambda \nu _{R}(d\lambda )}{k(z)\lambda -z}
\label{hk}
\end{equation}%
\begin{equation}
k(z)=\int_{0}^{\infty }\frac{\lambda \nu _{K}(d\lambda )}{h(z)\lambda +1},\;
\label{kh}
\end{equation}%
in which $\nu _{R}$ is defined in (\ref{nur}), $\nu _{K}$ is the probability
distribution of $(\varphi ^{\prime }(\sqrt{q}\gamma ))^{2}$ with \ $q$ of (%
\ref{qqn}) and the standard Gaussian random variable $\gamma $, i.e.,
\begin{equation}
\nu _{K}(\Delta )=\mathbf{P}\{(\varphi ^{\prime }(\sqrt{q}\gamma ))^{2}\in
\Delta \},\;\Delta \in \mathbb{R},  \label{nuka}
\end{equation}%
and we are looking for a solution of (\ref{hk}) -- (\ref{kh}) in the class
of pairs $(h,k)$ of functions such that $h$ is analytic outside the positive
semi-axis, continuous and positive on the negative semi-axis and
\begin{equation}
\Im h(z)\Im z>0,\;\Im z\neq 0;\;\sup_{\xi \geq 1}\xi h(-\xi )\in (0,\infty ).
\label{hcond}
\end{equation}
\end{theorem}


\begin{remark}
\label{r:rra} (i) To apply Theorem \ref{t:ind} to the proof of Theorem \ref%
{t:main} we need a version of Theorem \ref{t:ind} in which its "parameters",
i.e., $R_{n}$, hence $S_{n}$, in (\ref{mncal}) -- (\ref{nur}) and (possibly)
$\{x_{\alpha n}\}_{\alpha =1}^{n}$ in (\ref{Dn}) and (\ref{qqn}) are random,
defined for all $n$ on the same probability space $\Omega _{Rx}$,
independent of $\Omega _{Xb}$ of (\ref{oxb}) and satisfy conditions (\ref%
{r2}) -- (\ref{nur}) and (\ref{qqn}) with probability 1 on $\Omega _{Rx}$,
i.e., on a certain subspace $\overline{\Omega }_{Rx}\subset $ $\Omega
_{Rx},\;\mathbf{P}(\overline{\Omega _{Rx}})=1$. In this case Theorem \ref%
{t:ind} is valid with probability 1 in $\Omega _{Xb}\times \Omega _{Rx}$.
The corresponding argument is standard in random matrix theory (see, e.g.
Section 2.3 of \cite{Pa-Sh:11}) and similar to that presented in Remark \ref%
{r:rra1} (i). In deed, let $\overline{\Omega }_{Xb}(\omega _{Rx})\subset $ $%
\Omega _{Xb},\;\mathbf{P}(\overline{\Omega }_{Xb}(\omega _{Rx}))=1$ be the
subspace of $\Omega _{Xb}$ of (\ref{oxb}) on which the theorem holds for a
given realization $\omega _{Rx}\in \overline{\Omega }_{Rx}$ of the
parameters. Then it follows from the Fubini theorem that Theorem \ref{t:ind}
holds on a certain $\overline{\Omega }\subset \Omega _{Rx}\times \Omega
_{Xb},\;\mathbf{P}(\overline{\Omega })=1$. We will use this remark in the
proof of Theorem \ref{t:main}. The obtained limiting NCM\ $\nu _{\mathcal{M}%
} $ is random in general due to the randomness of $\nu _{R}$ and $q$ in (\ref%
{nur}) and (\ref{qqn}) which are defined on the probability space $\Omega
_{Rx}$ but do not depend on $\omega \in \overline{\Omega }_{Xb}$. We will
use this remark in the proof of Theorem \ref{t:main}. Note, however, that in
this case application the corresponding analogs of $\nu _{R}$ and $q$ are
not random, thus the limiting measure $\nu _{M^{L}}$ is a "genuine" \
non-random measure.

(ii) Repeating almost literally the proof of the theorem, one can treat a
more general case where $S_{m}$ is $m\times m$ positive definite matrix
satisfying (\ref{r2}) -- (\ref{nur}), $K_{n}$ is the $n\times n$ diagonal
matrix given by (\ref{Dn}) -- (\ref{qqn}, $X_{n}$ is a $n\times m$ Gaussian
random matrix satisfying (\ref{wga}) and (cf. (\ref{asf1})) $%
\lim_{m\rightarrow \infty ,n\rightarrow \infty }m/n=c\in (0,\infty ).$ The
corresponding modifications of the theorem are given in
Remark \ref{r:rra1} (ii).

(iii) The theorem is also valid for not necessarily Gaussian $X_{n}$ and $%
b_{n}$ (see \cite{Pa-Sl:20} and Remark \ref{r:penn}) (iii).
\end{remark}
We will prove now Theorem \ref{t:ind}

\smallskip
\begin{proof} Lemma \ref{l:mart} (i) implies that
the fluctuations of $\nu _{\mathcal{M}_{n}}$ vanish sufficiently fast as $%
n\rightarrow \infty $. This and the Borel-Cantelli lemma reduce the proof of
the theorem to the proof of the weak convergence of the expectation
\begin{equation}
\overline{\nu }_{\mathcal{M}_{n}}:=\mathbf{E}\{\nu _{\mathcal{M}_{n}}\}
\label{numc}
\end{equation}%
of $\nu _{\mathcal{M}_{n}}$ to the limit $\nu _{\mathcal{M}}$ whose
Stieltjes transform solves (\ref{fh}) -- (\ref{hcond}). It suffices to prove
the tightness of the sequence $\{\overline{\nu }_{\mathcal{M}_{n}}\}_{n}$ of
measures and the pointwise convergence on an open set of $\mathbb{C}%
\setminus \mathbb{R}_{+}$ of their Stieltjes transforms (cf. (\ref{stm}))
\begin{equation}
f_{\mathcal{M}_{n}}(z):=\int_{0}^{\infty }\frac{\overline{\nu }_{\mathcal{M}%
_{n}}(d\lambda )}{\lambda -z}  \label{stmn}
\end{equation}%
to the limit satisfying (\ref{fh}) -- (\ref{hcond}).

The tightness is guaranteed by the uniform in $n$ boundedness of%
\begin{equation}
\mu _{n}^{(1)}=\int_{0}^{\infty }\lambda \overline{\nu }_{\mathcal{M}%
_{n}}(d\lambda )  \label{mom}
\end{equation}%
providing the uniform in $n$ bounds for the tails of $\overline{\nu }_{%
\mathcal{M}_{n}}$.

According to the definition of the NCM (see, e.g. (\ref{ncm})), spectral
theorem and (\ref{mncal}) we have $\mu _{n}^{(1)}=\mathbf{E}\{n^{-1}\mathrm{%
Tr}\mathcal{M}_{n}\}=\mathbf{E}\{n^{-2}\mathrm{Tr}X_{n}R_{n}X_{n}^{T}K_{n}\}$
and then (\ref{tab}), (\ref{r2}) -- (\ref{Xn}) and (\ref{Dn}) -- (\ref{Phi})
yield%
\begin{equation}
\mu _{n}^{(1)}\leq n^{-2}\Phi _{1}^{2}\mathbf{E}\{\mathrm{Tr}%
X_{n}R_{n}X_{n}^{T}\}=\Phi _{1}^{2}n^{-1}\mathrm{Tr}R_{n}\leq
r_{2}^{1/2}\Phi _{1}^{2}.  \label{tim}
\end{equation}%
This implies the tightness of $\{\overline{\nu }_{\mathcal{M}_{n}}\}_{n}$
and reduces the proof of the theorem to the proof of pointwise in $\mathbb{C}%
\setminus \mathbb{R}_{+}$ convergence of (\ref{stmn}) to the limit
determined by (\ref{fh}) -- (\ref{kh}).

The above argument, reducing the analysis of the large size behavior of the
eigenvalue distribution of random matrices to that of the expectation of the
Stieltjes transform of the distribution, is widely used in random matrix
theory (see \cite{Pa-Sh:11}, Chapters 3, 7, 18 and 19), in particular, while
dealing with the sample covariance matrices. However, the matrix $\mathcal{M}%
_{n}$ of (\ref{mncal}) differs essentially from the sample covariance
matrices, since the "central" matrix $K_{n}$ of (\ref{Dn}) is random and
dependent on $X_{n}$ (data matrix according to statistics), while in the
sample covariance matrix the analog of $K_{n}$ is either deterministic or
random but independent of $X_{n}$.

This is why the next, in fact, the main step of the proof of Theorem \ref%
{t:ind} is to show that in the limit $n\rightarrow \infty $ the Stieltjes
transform (\ref{stmn}) of (\ref{mncal}) coincides with the Stieltjes
transform $f_{\mathsf{M}_{n}}$ of the mean NCM $\overline{\nu }_{\mathsf{M}%
_{n}}$of the matrix%
\begin{equation}
\mathsf{M}_{n}=S_{n}X_{n}^{T}\mathsf{K}_{n}X_{n}S_{n},  \label{mnbf}
\end{equation}%
where
\begin{equation}
\mathsf{K}_{n}=\{\delta _{jk}\mathsf{K}_{jn}\}_{j,k=1}^{n},\;\mathsf{K}%
_{jn}=(\varphi ^{\prime }(q_{n}^{1/2}\gamma _{j}))^{2},\;  \label{kab1}
\end{equation}%
$\varphi $ is again defined in (\ref{Dn}) -- (\ref{Phi}), $\{\gamma
_{j}\}_{j=1}^{n}$ are independent standard Gaussian random variables and $%
q_{n}$ is defined in (\ref{qqn}).

This, crucial for the paper fact, is proved in Lemma \ref{l:inter} below
provided that $\varphi $ in (\ref{Dn}) and (\ref{kab1}) and $S_{n}$, hence $%
R_{n}$ in (\ref{mncal}) and (\ref{mnbf}) satisfy the conditions%
\begin{equation}
\max_{x\in \mathbb{R}}|\varphi ^{(p)}(x)|=\widetilde{\Phi }_{p}<\infty
,\;p=0,1,2,  \label{phi3}
\end{equation}%
and
\begin{equation}
\sup_{n}||R_{n}||=\rho <\infty .  \label{rrho}
\end{equation}%
Thus, since $\mathsf{K}_{n}$, being random, is $X_{n}$-independent, the $%
n\rightarrow \infty $ limit of Stieltjes transform $f_{\mathsf{M}_{n}}$ of
the mean NCM $\overline{\nu }_{\mathsf{M}_{n}}$ of (\ref{mnbf}) can be
obtained by using one of the techniques of random matrix theory including
those of free probability theory \cite{Ch-Co:18,Mi-Sp:17} or based on the
Stieltjes transform, see \cite{Co-Ha:14,Pa-Sh:11} and references therein. We
will present below the corresponding assertion as Lemma \ref{l:rkx} and
outline its proof based on the Stieltjes transform techniques.

Hence, Lemmas \ref{l:inter} and \ref{l:rkx} imply that the limiting
Stieltjes transform $f_{\mathcal{M}}$ of (\ref{fh}) can be expressed via a
unique solution of the system (\ref{fh}) -- (\ref{hcond}), provided that $%
\varphi $ and $R_{n}$ in (\ref{mncal}) satisfy the conditions (\ref{phi3})
-- (\ref{rrho}), i.e., the assertion of Theorem \ref{t:ind} is proved under these conditions.
Let \ us show that these technical conditions can be replaced by initial
conditions (\ref{r2}) and (\ref{Phi}) of the theorem.

We will begin with (\ref{Phi}). For any $\varphi $ having a piecewise
continuous derivative and satisfying (\ref{Phi}) introduce%
\begin{eqnarray}
\varphi _{\varepsilon }(x) &=&(2\pi )^{-1/2}\int e^{-y^{2}/2}\varphi
(x+\varepsilon y)dy  \notag \\
&=&(2\pi \varepsilon ^{2})^{-1/2}\int e^{-(x-y)^{2}/2\varepsilon
^{2}}\varphi (y)dy,\;\varepsilon >0.  \label{phie}
\end{eqnarray}%
Then $\varphi _{\varepsilon }$ and $\varphi _{\varepsilon }^{\prime }$
converge to $\varphi $ and $\varphi ^{\prime }$ as $\varepsilon \rightarrow
0 $ uniformly on a compact set of $\mathbb{R}$ (except the discontinuity
points of $\varphi ^{\prime }$) and
\begin{equation}
\sup_{x\in \mathbb{R}}|\varphi _{\varepsilon }^{(p)}(x)|\leq \Phi
_{p},\;p=0,1,\;\sup_{x\in \mathbb{R}}|\varphi _{\varepsilon }^{^{\prime
\prime }}(x)|\leq \Phi _{1}/\varepsilon .  \label{phie0}
\end{equation}%
Hence, $\varphi _{\varepsilon }$ satisfies (\ref{phi3}) with $\widetilde{%
\Phi }_{p}=\Phi _{p},\;p=0,1$ and $\widetilde{\Phi }_{2}=\Phi
_{1}/\varepsilon <\infty $ and the assertion of theorem is valid for $%
\varphi _{\varepsilon }$ according to the above argument.

Let $\nu _{\mathcal{M}}$ be the measure whose Stieljes transform satisfies (%
\ref{fh}) -- (\ref{kh}) with $\nu _{R}$ such that \textrm{supp} $\nu
_{R}\subset \lbrack 0,\rho ],\;\rho <\infty $ (cf. (\ref{rrho})), $\varphi $
of (\ref{nuka}) be satisfying (\ref{Phi}), $\nu _{\mathcal{M}^{\varepsilon
}} $ be the analogous measure with $\varphi _{\varepsilon }$ instead of $%
\varphi $ in (\ref{nuka}), $\overline{\nu }_{\mathcal{M}_{n}}$ be the mean
NCM\ of (\ref{mncal}) and $\overline{\nu }_{\mathcal{M}_{n}^{\varepsilon }}$
be the mean NCM of the matrix (\ref{mncal}) with $\varphi _{\varepsilon }$
instead of $\varphi $ in (\ref{Dn}), i.e., with
\begin{equation}
K_{n}^{\varepsilon }=\{\delta _{jk}K_{jn}^{\varepsilon
}\}_{j,k=1}^{n},\;K_{jn}^{\varepsilon }=\left( \varphi _{\varepsilon
}^{\prime }\left( n^{-1/2}\sum_{\alpha =1}^{n}X_{j\alpha }x_{\alpha
n}+b_{j}\right) \right) ^{2},  \label{kaep}
\end{equation}%
instead of $K_{jn}$ of (\ref{Dn}). We write then for any $n$-independent $%
z\in \mathbb{C}\setminus \mathbb{R}_{+}$%
\begin{align}
&\hspace{-1cm}|f_{\mathcal{M}}(z)-f_{\mathcal{M}_{n}}(z)|\leq |f_{\mathcal{M}}(z)-f_{%
\mathcal{M}^{\varepsilon }}(z)|  \notag \\
&+|f_{\mathcal{M}^{\varepsilon }}(z)-f_{\mathcal{M}_{n}^{\varepsilon
}}(z)|+|f_{\mathcal{M}_{n}^{\varepsilon }}(z)-f_{\mathcal{M}_{n}}(z)|.
\label{3fm}
\end{align}%
According to Lemma \ref{l:uniq} (ii), the measure whose Stieltjes transform
solves (\ref{fh}) -- (\ref{kh}) is weakly continuous in $\nu _{K}$. Besides,
it follows from (\ref{nuka}) that $\nu _{K}$ is weakly continuous in $%
\varphi ^{\prime }$ with respect to the bounded point-wise convergence of $%
\varphi ^{\prime }$. Hence, the first term on the right of (\ref{3nu})
vanishes as $\varepsilon \rightarrow 0$. Next, the theorem proved above
under conditions (\ref{phi3}) -- (\ref{rrho}) implies that the second term
on the right vanishes as $n\rightarrow \infty $ for any $n$-independent $%
\varepsilon >0$. We conclude that the l.h.s. of (\ref{3fm}) vanishes as$%
\;n\rightarrow \infty $ if the third term on the right of (\ref{3fm})
vanishes as $\varepsilon \rightarrow 0$ uniformly in $n$:%
\begin{equation}
f_{\mathcal{M}_{n}^{\varepsilon }}(z)-f_{\mathcal{M}_{n}}(z)\rightarrow
0,\;\varepsilon \rightarrow 0,\;\zeta =\mathrm{dist}(z,\mathbb{R}_{+})\geq
\zeta _{0}>0.  \label{2fmn}
\end{equation}%
Denoting $\mathcal{G}=(\mathcal{M}_{n}-z)^{-1},\;\mathcal{G}_{\varepsilon }=(%
\mathcal{M}_{n}^{\varepsilon }-z)^{-1}$ and using the resolvent identity $%
\mathcal{G}_{\varepsilon }-\mathcal{G}=\mathcal{G}(\mathcal{M}_{n}-\mathcal{M%
}_{n}^{\varepsilon })\mathcal{G}_{\varepsilon }$ and the relations $f_{%
\mathcal{M}_{n}}(z)=\mathbf{E}\{n^{-1}\mathrm{Tr}\mathcal{G}\}$ and $f_{%
\mathcal{M}_{n}^{\varepsilon }}(z)=\mathbf{E}\{n^{-1}\mathrm{Tr}\mathcal{G}%
_{\varepsilon }\}$, we get
\begin{align}
&\hspace{-0.9cm}f_{\mathcal{M}_{n}^{\varepsilon }}(z)-f_{\mathcal{M}_{n}}(z) =n^{-1}%
\mathbf{E}\{\mathrm{Tr}\mathcal{G}_{\varepsilon }\mathcal{G}(\mathcal{M}_{n}-%
\mathcal{M}_{n}^{\varepsilon })\} \notag \\
&=n^{-2}\sum_{j=1}^{n}\mathbf{E}\{(XS\mathcal{G}_{\varepsilon }\mathcal{G}%
SX^{T})_{jj}(K_{jn}-K_{jn}^{\varepsilon })\}.\label{fmfe}
\end{align}%
Now, (\ref{rrho}), Schwarz inequality for expectations and the bounds
\begin{equation}
|K_{j}|\leq \Phi _{1}^{2},\;||\mathcal{G}||\leq \zeta ^{-1},\;||\mathcal{G}%
_{\varepsilon }||\leq \zeta ^{-1},\;\zeta =\mathrm{dist}\{z,\mathbb{R}%
_{+}\}\geq \zeta _{0}>0,  \label{kgbo}
\end{equation}%
where we used the bound
\begin{equation}
||(A-z)^{-1}||\leq \zeta ^{-1}  \label{Gb}
\end{equation}%
valid for any positive definite $A$, yield for the r.h.s. of (\ref{fmfe})%
\begin{align*}
& \hspace{-1.5cm}\rho (\zeta n)^{-2}\sum_{j=1}^{n}\mathbf{E}%
\{||X^{(j)}||^{2}|K_{jn}-K_{jn}^{\varepsilon })|\} \\
& \leq \rho (\zeta n)^{-2}\sum_{j=1}^{n}\mathbf{E}^{1/2}\{||X^{(j)}||^{4}\}%
\mathbf{E}^{1/2}\{|K_{jn}-K_{jn}^{\varepsilon }|^{2}\},
\end{align*}%
where $X^{(j)}=\{X_{j\alpha }\}_{\alpha =1}^{n},\;j=1,...,n$ are the columns
of the $n\times n$ matrix $X$. \ Taking into account that%
\begin{equation}
||X^{(j)}||^{2}=\sum_{\alpha =1}^{n}X_{j\alpha }^{2}  \label{nx2}
\end{equation}%
and that $\{X_{j\alpha }\}_{\alpha =1}^{n}$ are independent standard
Gaussian (see (\ref{wga})), we obtain%
\begin{equation}
\mathbf{E}\{||X^{(j)}||^{2}\}=n,\;\mathbf{E}\{||X^{(j)}||^{4}\}=n(n+2)\leq
Cn^{2},\;C\geq 3.  \label{exj4}
\end{equation}%
Since, in addition, $\{(K_{jn}-K_{jn}^{\varepsilon })\}_{j=1}^{n}$ are
i.i.d. random \ variables, we have in view of (\ref{Dn}), (\ref{phie}) and (%
\ref{exj4}):
\begin{align*}
& \hspace{-0.5cm}|f_{\mathcal{M}_{n}^{\varepsilon }}(z)-f_{\mathcal{M}%
_{n}}(z)|\leq C^{1/2}\rho \zeta ^{-2}\mathbf{E}^{1/2}\{|K_{1n}-K_{1n}^{%
\varepsilon }|^{2}\} \\
& \leq C^{1/2}\rho \zeta ^{-2}((2\pi )^{-1/2}\int e^{-y^{2}/2}|\varphi
^{\prime }(x)-\varphi ^{\prime }(x+\varepsilon y)|^{2}\Gamma
_{n}(dx)dy)^{1/2},
\end{align*}%
where $\Gamma _{n}$ is the probability law of the argument of $\varphi
^{\prime }$ in (\ref{Dn}) and (\ref{kaep}). Since $\{X_{j\alpha
}\}_{j,\alpha =1}^{n}$ and $\{b_{j}\}_{j=1}^{n}$ are independent standard
Gaussian, $\Gamma _{n}(dx)=g_{n}(x)dx$, where $g_{n}$ is the density of the
Gaussian distribution of zero mean and variance $q_{n}$ of (\ref{qqn}), the
r.h.s. of the above expression tends to zero as $\varepsilon \rightarrow 0$
uniformly in $n\rightarrow \infty $. This proves (\ref{2fmn}), hence,
justifies the replacement of (\ref{phi3}) by the condition (\ref{Phi}) of
the theorem.

Next, we will replace (\ref{rrho}) by condition of (\ref{r2}) of the
theorem. This is, in fact, a known procedure of random matrix theory. In our
case it is a version of the procedure given in the first part of proof of
Theorem 7.2.2 (or Theorem 19.1) in \cite{Pa-Sh:11}. Here is an outline of
the procedure. Let $R_{n}$ be a general (i.e., not satisfying in general (%
\ref{rrho})) positive definite matrix such that (\ref{r2}) -- (\ref{nur})
hold with certain $r_{2}$ and the limiting measure $\nu _{R}$. For any
positive integer $p$ introduce the truncated matrix $R_{n}^{(p)}$ having the
same eigenvectors as $R_{n}$ and  eigenvalues
$R_{\alpha }^{(p)}=\max \{R_{\alpha },p\}$,$%
\;\alpha =1,2,...,n$, where $\{R_{\alpha}\}_{\alpha=1}^n$ are the eigenvalues of $R_n$. Then $R_{n}^{(p)}$ satisfies (\ref{rrho}) with $\rho
=p $, its NCM $\nu _{R_{n}^{(p)}}$ satisfies (\ref{r2}) -- (\ref{nur}) with
the weak limit $\nu _{R^{(p)}}:=\lim_{n\rightarrow \infty }\nu
_{R_{n}^{(p)}} $ coinciding with $\nu _{R}$ inside $[0,p)$, equals zero
outside $[0,p]$ and such that
\begin{equation}
\lim_{p\rightarrow \infty }\nu _{R^{(p)}}=\nu _{R}.  \label{lpnr}
\end{equation}%
Denote by $\mathcal{M}_{n}^{(p)}$ the matrix (\ref{mncal}) with $R_{n}^{(p)}$
instead of $R_{n}$, by $\overline{\nu }_{\mathcal{M}_{n}^{(p)}}$ its mean
NCM and by $\nu _{\mathcal{M}^{(p)}}$ its limit as $n\rightarrow \infty $
with a fixed $p>0$. We will use now an argument analogous to that used above
to prove the replacement of (\ref{phi3}) by (\ref{Phi}). We write (cf. (\ref%
{3fm}))
\begin{equation}
|\nu _{\mathcal{M}}-\overline{\nu }_{\mathcal{M}_{n}}|\leq |\nu _{\mathcal{M}%
}-\nu _{\mathcal{M}^{(p)}}|+|\nu _{\mathcal{M}^{(p)}}-\overline{\nu }_{%
\mathcal{M}_{n}^{(p)}}|+|\overline{\nu }_{\mathcal{M}_{n}^{(p)}}-\overline{%
\nu }_{\mathcal{M}_{n}}|.  \label{3nu}
\end{equation}%
It follows then from Lemma \ref{l:uniq} (ii) and (\ref{lpnr}) that solution
of (\ref{hk}) -- (\ref{kh}), hence (\ref{fh}), with $\nu _{R^{(p)}}$ instead
$\nu _{R}$ converges pointwise in $\mathbb{C}\setminus \mathbb{R}_{+}$ as $%
p\rightarrow \infty $ to that of (\ref{fh}) -- (\ref{kh}) with the "genuine"
$\nu _{R}$ satisfying (\ref{r2}) (see also (\ref{nukr})). Thus, the first
term on the right vanishes as $p\rightarrow \infty $. Next, since the
theorem is valid under condition (\ref{rrho}), hence (\ref{r2}) -- (\ref{nur}%
), and $R_{n}^{(p)}$ satisfies (\ref{rrho}) with $\rho =p$, the second term
on the right vanishes as $n\rightarrow \infty $ for any $n$-independent $p>0$%
. Thus, it suffices to prove that
\begin{equation*}
\overline{\nu }_{\mathcal{M}_{n}}-\overline{\nu }_{\mathcal{M}_{n}^{(p)}}
\end{equation*}%
tends weakly to zero as $p\rightarrow \infty $ uniformly in $n\rightarrow
\infty $ (cf. (\ref{2fmn})). The expectations $\overline{\nu }_{\mathcal{M}%
_{n}}$ and $\overline{\nu }_{\mathcal{M}_{n}^{(p)}}$ coincide with those $%
\overline{\nu }_{M_{n}}$ and $\overline{\nu }_{M_{n}^{(p)}}$ of matrices $%
M_{n}=D_{n}X_{n}R_{n}X_{n}^{T}D_{n}$ and $%
M_{n}^{(p)}=D_{n}X_{n}R_{n}^{(p)}X_{n}^{T}D_{n}$ (cf. (\ref{JJM}). Writing $%
M_{n}$ as the sum of the rank-one matrices (cf. (\ref{JJM}) and (\ref{cmyy}))%
\begin{equation}
M_{n}=\sum_{\alpha =1}^{n}Y_{\alpha }\otimes Y_{\alpha },\;Y_{\alpha
}=\{Y_{j\alpha }\}_{j=1}^{n},\;Y_{j\alpha }=(D_{n}X_{n}S_{n})_{j\alpha }
\label{nmyy}
\end{equation}%
and using the analogous representation for $M_{n}^{(p)}$, we conclude that
\begin{equation*}
\mathrm{rank}(M_{n}-M_{n}^{(p)})\leq \sharp \{R_{\alpha }:R_{\alpha
}>p,\;\alpha =1,2,...,n\}
\end{equation*}%
and then the min-max principle of linear algebra and the definition of a NCM
(see, e.g. (\ref{ncm})) yield for any interval $\Delta $ of spectral axis%
\begin{equation}
|\overline{\nu }_{\mathcal{M}_{n}}(\Delta )-\overline{\nu }_{\mathcal{M}%
_{n}^{(p)}}(\Delta )|\leq \nu _{R_{n}}([p,\infty )).  \label{nmmr}
\end{equation}%
This estimate and (\ref{nur}) imply the weak convergence of the r.h.s.
to zero as $p\rightarrow \infty $ uniformly in $n$, hence, the weak \
convergence of $\overline{\nu }_{\mathcal{M}_{n}}$ to $\nu _{\mathcal{M}}$
as $n\rightarrow \infty $ and the coincidence of the Stieltjes transform of $%
\nu _{\mathcal{M}}$ with that given by (\ref{fh}) -- (\ref{hk}) under
condition (\ref{r2}).
\end{proof}

We will prove now an assertion which is used in the proof of the theorem and
which is central in this work since it shows the mathematical mechanism of
the coincidence of the limiting eigenvalue distribution of "non-linear"
random matrix $\mathcal{M}_{n}$ of (\ref{mncal}), where $K_{n}$ of (\ref{Dn}%
) depends nonlinearly on $X_{n}$, and a conventional for random matrix
theory matrix $\mathsf{M}_{n}$ of (\ref{mnbf}), where the analog $\mathsf{K}%
_{n}$ of $K_{n}$ is random but independent of $X_{n}$ matrix given by (\ref%
{kab1}).

\begin{lemma}
\label{l:inter} Consider the matrices $\mathcal{M}_{n}$ and $\mathsf{M}_{n}$
given by (\ref{mncal}) and (\ref{mnbf}) and such that:

- the matrix $S_{n}$ in $\mathcal{M}_{n}$ and $\mathsf{M}_{n}$ is diagonal,
positive definite and satisfies (\ref{rrho});

- the random matrix $X_{n}$ in $\mathcal{M}_{n}$ and $\mathsf{M}_{n}$ is
Gaussian and given by (\ref{Xn});

- the matrix $K_{n}$ in $\mathcal{M}_{n}$ is defined in (\ref{Dn}) with $%
\varphi $ satisfying (\ref{phi3});

- the matrix $\mathsf{K}_{n}$ in $\mathsf{M}_{n}$ is defined in (\ref{kab1})
with the same $\varphi $ satisfying (\ref{phi3}).

\noindent Denote by $\overline{\nu }_{\mathcal{M}_{n}}$ and $\overline{\nu }%
_{\mathsf{M}_{n}}$ the mean NCM of $\mathcal{M}_{n}$ and $\mathsf{M}_{n}$,by
$f_{\mathcal{M}_{n}}$ and $f_{\mathsf{M}_{n}}$ the Stieltjes transforms of $%
\overline{\nu }_{\mathcal{M}_{n}}$ and $\overline{\nu }_{\mathsf{M}_{n}}$ and%
\begin{equation}
\Delta _{n}(z):=f_{\mathcal{M}_{n}}(z)-f_{\mathsf{M}_{n}}(z),\;z\in \mathbb{C%
}\setminus \mathbb{R}_{+}.  \label{dfmm}
\end{equation}

Then we have for any $n$-independent $z,\;\zeta :=\mathrm{dist}\{z,\mathbb{R}%
_{+}\}>0:$%
\begin{equation}
\lim_{n\rightarrow \infty }\Delta _{n}(z)=0.  \label{de0}
\end{equation}
\end{lemma}

\begin{proof}
Writing
\begin{equation}\label{fGE}
\ f_{\mathcal{M}_{n}}=\mathbf{E\{}n^{-1}\mathrm{Tr}\mathcal{G}_{n}(z)\}, \;
f_{\mathsf{M}_{n}}=\mathbf{E}\{n^{-1}\mathrm{Tr}\mathsf{G}_{n}(z)\}
\end{equation}
where
\begin{equation}
\mathcal{G}_{n}(z)=(\mathcal{M}_{n}-z)^{-1},\;\mathsf{G}_{n}(z)=(\mathsf{M}%
_{n}-z)^{-1},\;z\in \mathbb{C}\setminus \mathbb{R}_{+}  \label{gcgb}
\end{equation}%
are the corresponding resolvents, we obtain from (\ref{dfmm})
\begin{equation}
\Delta _{n}(z)=\mathbf{E\{}n^{-1}\mathrm{Tr(}\mathcal{G}_{n}(z)- \mathsf{G}%
_{n}(z))\}.  \label{de}
\end{equation}%
Note that the symbol $\mathbf{E\{...\}}$ in (\ref{de}) and below denotes the
expectation with respect to the "old" collections $\{X_{j\alpha
}\}_{j,\alpha =1}^{n}$ and$\;\{b_{j}\}_{j=1}^{n}$ of (\ref{Xn}) and (\ref{b}%
) as well as with respect to the "new" collection $\{\gamma _{j}\}_{j=1}^{n}$
of (\ref{kab1}) of independent standard Gaussian variables.

Set for $j=1,...,n$%
\begin{align}
&\eta _{j}(t)=t^{1/2}\eta _{j}+(1-t)^{1/2}q_{n}^{1/2}\gamma _{j},
\;t \in [0,1], \notag
\\&\eta _{j} =n^{-1/2}\sum_{\alpha =1}^{n}X_{j\alpha }x_{\alpha n}+b_{j},
\label{eetaj}
\end{align}%
\begin{equation}
K_{n}(t)=\{\delta _{jk}K_{jn}(t)\}_{j,k=1}^{n},\;K_{jn}(t)=(\varphi ^{\prime
}(\eta _{j}(t)))^{2}  \label{kat}
\end{equation}%
and%
\begin{equation}
\mathcal{M}_{n}(t)=S_{n}X_{n}^{T}K_{n}(t)X_{n}S_{n},\;\mathcal{G}_{n}(z,t)=(%
\mathcal{M}_{n}(t)-z)^{-1}.  \label{mtgt}
\end{equation}%
Then $\mathcal{M}_{n}(1)=\mathcal{M}_{n},\;\mathcal{M}_{n}(0)=\mathsf{M}_{n}$
and by using the formula
\begin{equation}
\frac{d}{dt}A^{-1}(t)=-A^{-1}(t)\frac{d}{dt}A(t)A^{-1}(t),  \label{derr}
\end{equation}%
valid for any matrix function $A$ invertible uniformly in $t$, we obtain in
view of (\ref{gcgb}) and (\ref{kat}):%
\begin{equation*}
\Delta _{n}(z)=\frac{1}{n}\int_{0}^{1}\frac{d}{dt}\mathbf{E}\{\mathrm{Tr}%
\mathcal{G}_{n}(z,t)\}dt=-\frac{1}{n}\int_{0}^{1}\mathbf{E\{}\mathrm{Tr}%
\mathcal{G}_{n}^{2}(z,t)\dot{\mathcal{M}}(t)\}dt,
\end{equation*}%
where
\begin{align*}
& \hspace{2cm}\dot{\mathcal{M}}_{n}(t)=\frac{d}{dt}\mathcal{M}_{n}(t)=:\{%
\dot{\mathcal{M}}_{\alpha \beta }(t)\}_{\alpha ,\beta =1}^{n}, \\
& \dot{\mathcal{M}}_{\alpha \beta }(t)=\frac{1}{2n}%
\sum_{j=1}^{n}(S_{n}X_{n}^{T}\dot{K}_{n}(t))_{\alpha j}(t^{-1/2}\eta
_{j}-(1-t)^{-1/2}q_{n}^{1/2}\gamma _{j})(X_{n}S_{n})_{j\beta },
\end{align*}%
and according to (\ref{kab1}) -- (\ref{phi3}), and (\ref{kat})%
\begin{equation}
\dot{K}_{jn}(t)=2(\varphi ^{\prime }\varphi ^{\prime \prime })(x)|_{x=\eta
_{j}(t)}.  \label{Kd}
\end{equation}%
By using (\ref{derr}) again, we get%
\begin{align}
\Delta _{n}(z)& =\delta _{n}^{\prime }(z), \notag\\
\delta _{n}(z)
& =\frac{1}{2n^{2}}\sum_{j=1}^{n}\int_{0}^{1}\mathbf{E}\{F_{j}(z,t)(t^{-1/2}%
\eta _{j}-(1-t)^{-1/2}q_n^{1/2}\gamma _{j})\}t^{-1/2}dt,  \label{De1}
\end{align}%
where%
\begin{equation}
F_{j}(z,t)=(X_{n}S_{n}\mathcal{G}_{n}(z,t)S_{n}X_{n}^{T}\dot{K}_{n}(t))_{jj}.
\label{Fj}
\end{equation}%
It suffices to prove that%
\begin{equation}
\max_{z\in O}|\delta _{n}(z)|=o(1),\;n\rightarrow \infty ,  \label{den}
\end{equation}%
where $O$ is an open set lying strictly inside $\mathbb{C}\setminus \mathbb{R%
}_{+}$. Indeed, since $F_{j}$ is analytic in $\mathbb{C}\setminus \mathbb{R}%
_{+}$, $\delta _{n}$ is analytic in $O$ and any such bound implies (\ref{de0}%
) by the Cauchy theorem.

To deal with the expectation in\ \ the r.h.s. of the second equality in (\ref%
{De1}), we take into account that $\{X_{j\alpha }\}_{\alpha =1}^{n}$ and $%
\gamma _{j}$ are independent Gaussian random variables (see (\ref{Xn})) and (%
\ref{kab1})) and use the simple differentiation formula%
\begin{equation}
\mathbf{E}\{\xi f(\xi )\}=\mathbf{E}\{f^{\prime }(\xi )\}  \label{difg}
\end{equation}%
valid for the standard Gaussian random variable and any differentiable $f:$ $%
\mathbb{R\rightarrow C}$ with a polynomially bounded derivative.

The formula, applied to $\eta _{j}$'s and $\gamma _{j}$'s in the integrand
of (\ref{De1}), yields%
\begin{align}
&\hspace{-1cm}\mathbf{E}\{F_{j}(z,t)(t^{-1/2}\eta _{j}-(1-t)^{-1/2}q_n^{1/2}\gamma
_{j})\}\notag
\\&\hspace{2cm}=(tn)^{-1/2}\sum_{\alpha =1}^{n}\mathbf{E}\left\{ \frac{\partial F_{j}%
}{\partial X_{j\alpha }}\right\} x_{\alpha n},\label{lin1}
\end{align}%
where the partial derivative in the r.h.s. denotes the "explicit" derivative
(not applicable to $X_{j\alpha }$ in the argument of $K_{jn}$ and $\dot{K}%
_{jn}$ of (\ref{kat})).

By using the formula%
\begin{equation}
\frac{\partial \mathcal{G}_{\beta \gamma }}{\partial X_{j\alpha }}=-\frac{1}{%
n}(\mathcal{G}SX^{T}K)_{\beta j}(S\mathcal{G})_{\alpha \gamma }-\frac{1}{n}(%
\mathcal{G}S)_{\beta \alpha }(KXS\mathcal{G})_{j\gamma },  \label{dgdx}
\end{equation}%
which follows from (\ref{derr}) and where we omitted the subindex $n$ in all
the matrices and denoted $\mathcal{G}=\mathcal{G}_{n}(z,t),\;K=K_{n}(t)$
(see (\ref{mtgt}) and (\ref{kat})), we obtain%
\begin{align}
&(tn)^{-1/2}\sum_{\alpha =1}^{n}\mathbf{E}\left\{ \frac{\partial F_{j}}{%
\partial X_{j\alpha }}\right\} x_{\alpha n} \notag
\\&\hspace{1cm}=\frac{2}{(tn)^{1/2}}\mathbf{E}\{(%
\dot{K}X\mathcal{G}_{S}x)_{j}(1-n^{-1}(KX\mathcal{G}_{S}X^{T})_{jj})\},
\label{lin2}
\end{align}%
where
\begin{equation}
\mathcal{G}_{S}=S\mathcal{G}S.  \label{gs}
\end{equation}%
We have then via (\ref{phi3}), (\ref{rrho}), (\ref{kgbo}), (\ref{kat}) and (\ref{Kd})
\begin{equation}
|K_{j}|\leq \widetilde{\Phi }_{1}^{2},\;|\dot{K}_{j}|\leq 2\widetilde{\Phi }%
_{1}\widetilde{\Phi }_{2},\;||\mathcal{G}_{S}||\leq \rho \zeta ^{-1}.
\label{kgbo1}
\end{equation}%
This and (\ref{exj4}) imply that the r.h.s. of (\ref{lin2}) admits the bounds%
\begin{align}
& \hspace{-1.5cm}\frac{4\widetilde{\Phi }_{1}\widetilde{\Phi }_{2}\rho }{%
\zeta (tn)^{1/2}}||x||\mathbf{E\{}||X^{(j)}||+(\widetilde{\Phi }%
_{1})^{2}\rho (\zeta n)^{-1}||X^{(j)}||^{3}\}  \notag \\
& \leq \frac{4\widetilde{\Phi }_{1}\widetilde{\Phi }_{2}\rho }{\zeta t^{1/2}}%
||x||\mathbf{(}1+C^{3/4}\rho \zeta ^{-1}(\widetilde{\Phi }_{1})^{2}),\;
\label{bou1}
\end{align}%
and since, according to (\ref{qqn}), $||x||=O(n^{1/2})$, we combine the
above bound with (\ref{lin1}) and (\ref{lin2}) to conclude that the
expectation in the r.h.s. of (\ref{De1}) is $\varepsilon _{n}(z,t)$, where $%
\varepsilon _{n}(z,t)=O(n^{1/2}),\;n\rightarrow \infty $ uniformly in $t\in
\lbrack 0,1]$ and $z$ belonging to an open set $O$ lying strictly inside $%
\mathbb{C}\setminus \mathbb{R}_{+}$.

We have proved (\ref{den}), hence (\ref{de0}), both with the r.h.s. of the
order $O(n^{-1/2})$ uniformly in $z\in O\subset \mathbb{C}\setminus \mathbb{R%
}_{+}$.
\end{proof}

\begin{remark}
\label{r:inter} The "interpolating" random variable (\ref{eetaj}) implements
a simple version of the "interpolation" procedure used in \cite{Pa:00},
Theorem 5.7 and in  \cite{Pa-Sh:11}, Sections 18.3 - 18.4 and 19.1 - 19.2
to pass from the Gaussian random matrices to matrices with i.i.d. entries.
The procedure can be viewed as a manifestation of the so-called Lindeberg
principle, see \cite{Go-Co:15} for related results and references.
\end{remark}

We will find now the limiting eigenvalue distribution of a class of random
matrices containing (\ref{mnbf}) and used in the proof of Theorem \ref{t:ind}%
. In particular, we obtain functional equations (\ref{fh}) -- (\ref{kh})
determining uniquely the Stieltjes transform of the distribution, hence, the
distribution as well as the operation "$\diamond$" in (\ref{opdia}). We will use for these, more general, matrices the same
notation $\mathsf{M}_{n}$. Note that we give here a rather simple version of
the assertion sufficient to prove Theorem \ref{t:ind}. For more general versions
see, e.g. \cite{Co-Ha:14} and references therein.

\begin{lemma}
\label{l:rkx} Consider the $n\times n$ random matrix
\begin{equation}
\mathsf{M}_{n}=n^{-1}S_{n}X_{n}^{T}\mathsf{K}_{n}X_{n}S_{n},  \label{mnsf}
\end{equation}%
(see (\ref{mnbf}) -- (\ref{kab1})), where $S_{n}$ satisfies (\ref{r2}) and \
(\ref{nur}), $X_{n}$ has standard Gaussian entries (see (\ref{Xn})) and $%
\mathsf{K}_{n}$ is a $\ n\times n$ positive definite matrix such that (cf. (%
\ref{r2}) -- (\ref{nur}))
\begin{equation}
\sup_{n}n^{-1}\mathrm{Tr}\mathsf{K}_{n}^{2}\leq k_{2}<\infty ,  \label{k21}
\end{equation}%
\begin{equation}
\lim_{n\rightarrow \infty }\nu _{\mathsf{K}_{n}}=\nu _{\mathsf{K}},\;\nu _{%
\mathsf{K}}(\mathbb{R})=1,  \label{kur1}
\end{equation}%
where $\nu _{\mathsf{K}_{n}}$ is the Normalized Counting Measure of $\mathsf{%
K}_{n}$, $\nu _{\mathsf{K}}$ is a non-negative and not concentrated ar zero
measure (cf. (\ref{r2}) -- (\ref{nur})) and $\lim $ denotes the weak
convergence of probability measures.

Then the Normalized Counting Measure $\nu _{\mathsf{M}_{n}}$ of $\ \mathsf{M}%
_{n}$ converges weakly with probability 1 to a non-random measure $\nu _{%
\mathsf{M}},\;\nu _{\mathsf{M}}(\mathbb{R}_{+})=1$ and its Stieltjes
transform $f_{\mathsf{M}}$ (see (\ref{stm})) can be obtained from the system
(\ref{hk}) -- (\ref{kh}) in which $\nu _{K}$ is replaced by $\nu _{\mathsf{K}%
}$ of (\ref{k21}) -- (\ref{kur1}) and which is uniquely solvable in the
class of pairs $(h,k)$ of functions such that $h$ is analytic outside the
positive semi-axis, continuous and positive on the negative semi-axis and
satisfies (\ref{hcond}).
\end{lemma}


\begin{remark}
\label{r:rra1} (i) It is easy to check that the assertions of the lemma
remain valid with probability 1 in the case where the "parameters" of the
theorem, i.e., $S_{n}$, hence $R_{n}$, in (\ref{r2}) -- (\ref{nur}) and $%
\mathsf{K}_{n}$ (\ref{k21}) -- (\ref{kur1}) are random, defined for all $n$
on the same probability space $\Omega _{R\mathsf{K}}$, independent of $%
X_{n}=\{X_{j\alpha }\}_{j,\alpha =1}^{n}$ for every $n$ and satisfies
conditions (\ref{r2}) -- (\ref{nur}) and (\ref{k21}) -- (\ref{kur1}) with
probability 1 on $\Omega _{R\mathsf{K}}$, i.e., on a certain subspace $%
\overline{\Omega }_{R\mathsf{K}}\subset $ $\Omega _{R\mathsf{K}},\;\mathbf{P}%
(\overline{\Omega }_{R\mathsf{K}})=1$. This follows from an argument
analogous to that presented in Remark \ref{r:rra} (i). In this case $\mathbf{%
E}\{...\}$ denotes the expectation with respect to $X_{n}$.

\smallskip (ii) Repeating almost literally the proof of the lemma, one can
treat a more general case where $S_{m}$ is a $m\times m$ positive definite
matrix satisfying (\ref{r2}) -- (\ref{nur}), $K_{n}$ is a $n\times n$
positive definite matrix satisfying (\ref{k21}) -- (\ref{kur1}), $X_{n}$ is
a $n\times m$ Gaussian random matrix satisfying (\ref{wga}) and (cf. (\ref%
{asf1}))
\begin{equation}
\lim_{m\rightarrow \infty ,n\rightarrow \infty }m/n=c\in (0,\infty ).
\label{mpl1}
\end{equation}%
In this case the Stieltjes transform $f_{\mathsf{M}}$ of the limiting NCM is
again uniquely determined by three functional equations, where the first and
the third coincide with (\ref{fh}) and (\ref{kh}) while the second is (\ref%
{hk}) in which $k(z)$ is replaced by $k(z)c^{-1}$ (see, e.g. \cite{Co-Ha:14}%
) and references therein.

\smallskip (iii) The lemma is also valid for not necessarily Gaussian $X_{n}$
(see \cite{Co-Ha:14,Pa-Sl:20} and references therein for more general cases
of the theorem and their properties. If, however, we confine ourselves to
the Gaussian case, then we can reformulate our result in terms of correlated
Gaussian entries. Indeed, let $Z_{n}=\{Z_{j\alpha }\}_{j,\alpha =1}^{n}$ be
a Gaussian matrix with
\begin{equation*}
\mathbf{E}\{Z_{j\alpha }\}=0,\;\mathbf{E}\{Z_{j_{1}\alpha
_{1}}Z_{j_{2}\alpha _{2}}\}=C_{j_{1}\alpha _{1},j_{2}\alpha _{2}},
\end{equation*}%
and a separable covariance matrix $C_{j_{1}\alpha _{1},j_{2}\alpha _{2}}=%
\mathsf{K}_{j_{1}j_{2}}R_{\alpha _{1}\alpha _{2}}$, i.e., $C=\mathsf{K}%
\otimes R$ and $\mathsf{K}_{n}=\{\mathsf{K}_{j_{1}j_{2}}%
\}_{j_{1},j_{2}=1}^{n}$ and $R_{n}=\{R_{\alpha _{1}\alpha _{2}}\}$ as in the
lemma. Writing $\mathsf{K}_{n}=\mathsf{D}_{n}^{2},\;R_{n}=S_{n}^{2}$ and
denoting $Z_{n}=S_{n}X_{n}D_{n}$, we can view as a data matrix and then the
corresponding sample covariance matrix $Z_{n}^{T}Z_{n}$ is (\ref{mnsf}) of
spatial-temporal correlated time series.
\end{remark}
Here is the proof of Lemma \ref{l:rkx}.

\smallskip
\begin{proof}
As it was in the proof of Theorem \ref%
{t:ind}, Lemma \ref{l:mart} (i) together with the Borel-Cantelli lemma
reduce the proof of the theorem to that of the weak convergence of the
expectation
\begin{equation}
\overline{\nu }_{\mathsf{M}_{n}}=\mathbf{E}\{\nu _{\mathsf{M}_{n}}\}
\label{numsf}
\end{equation}%
of $\nu _{\mathsf{M}_{n}}$.

Next, it follows from the condition of the lemma that the argument analogous
to that proving (\ref{tim}) yields
\begin{equation*}
\int_{0}^{\infty }\lambda \overline{\nu }_{\mathsf{M}_{n}}(d\lambda )\leq
k_{2}^{1/2}r_{2}^{1/2}<\infty ,
\end{equation*}%
hence, the tightness of measures $\{\overline{\nu }_{\mathsf{M}_{n}}\}_{n}$
and, in turn, reduces the proof of the lemma to that of the pointwise
convergence in $\mathbb{C}\setminus \mathbb{R}_{+}$ of their Stieltjes
transforms $f_{\mathsf{M}_{n}}$ to the limit $f$ satisfying (\ref{fh}) -- (%
\ref{kh}). Moreover, the analyticity of $f_{\mathsf{M}_{n}},\;f_{\mathsf{M}%
},\;h$ and $k$ in $\mathbb{C}\setminus \mathbb{R}_{+}$ (see Lemma \ref%
{l:uniq}) allows us to confine ourselves to the open negative semi-axis
\begin{equation}
I_{-}=\{z\in \mathbb{C}:z=-\xi ,\;0<\xi <\infty \}.  \label{imin}
\end{equation}%
Thus, we will mean and often write explicitly below that $z\in I_{-}$.
%

Note first that since $\{X_{j\alpha }\}_{j,\alpha =1}^{n}$ are standard
Gaussian, we can assume without loss of generality that $S_{n}$ and $\mathsf{%
K}_{n}$ are diagonal, i.e.,
\begin{equation}
S_{n}=\{\delta _{\alpha \beta }S_{\alpha n}\}_{\alpha ,\beta =1}^{n},\;%
\mathsf{K}_{n}=\{\delta _{jk}\mathsf{K}_{jn}\}_{j,k=1}^{n}.  \label{snkn}
\end{equation}%
Given $j=1,...,n,$ consider the $n\times n$ matrix
\begin{equation}
H^{(j)}=\{H_{\alpha \beta }^{(j)}\}_{\alpha ,\beta =1}^{n},\;H_{\alpha \beta
}^{(j)}:=n^{-1}(\mathsf{G}SX^{T})_{\alpha j}{}(\mathsf{K}XS)_{j\beta }
\label{H}
\end{equation}%
and we omit here and below the subindex $n$ in the notation of matrices and
their entries.

It follows from (\ref{fGE}) -- (\ref{gcgb}) and the resolvent identity%
\begin{equation}
\mathsf{G}=-z^{-1}+z^{-1}\mathsf{GM},  \label{rid}
\end{equation}%
implying
\begin{equation}
z^{-1}\sum_{j=1}^{n}\mathbf{E}\{H^{(j)}\}=z^{-1}\mathbf{E}\{\mathsf{GM}\},
\label{hgm}
\end{equation}%
that it suffices to find the $n\rightarrow \infty $ limit of $\mathbf{E}%
\{n^{-1}\mathrm{Tr} H^{(j)}\}$.

To this end we will apply to the expectation in the r.h.s. of (\ref{H}) the
Gaussian differentiation formula (\ref{difg}). We compute the derivative of $%
\mathsf{G}_{\alpha \gamma }$ with respect to $X_{j\gamma }$ by using an
analog of (\ref{dgdx}) and we obtain
\begin{eqnarray}
\mathbf{E}\{H_{\alpha \beta }^{(j)}\} &=&n^{-1}\mathbf{E}\{\mathsf{G}%
_{\alpha \beta }\}S_{\beta }^{2}\mathsf{K}_{j}-\mathbf{E}\{h_{n}(z)H_{\alpha
\beta }^{(j)}\}\mathsf{K}_{j}  \label{eh1} \\
&&-n^{-2}\mathbf{E}\{(\mathsf{G}S^{2}\mathsf{G}SX^{T})_{\alpha
j}(XS)_{j\beta }\}\mathsf{K}_{j}^{2},  \notag
\end{eqnarray}%
where%
\begin{equation}
h_{n}(z)=n^{-1}\mathrm{Tr}S\mathsf{G}S=n^{-1}\mathrm{Tr}\mathsf{G}%
R,\;R=S^{2}.  \label{hR}
\end{equation}%
We write
\begin{equation}
h_{n}=\overline{h}_{n}+(h_{n}-\overline{h}_{n}),\;\overline{h}_{n}=\mathbf{E}%
\{h_{n}\}  \label{hnb}
\end{equation}%
in the r.h.s. of (\ref{eh1}) and get%
\begin{eqnarray}
\mathbf{E}\{H_{\alpha \beta }^{(j)}\} &=&n^{-1}\mathbf{E}\{\mathsf{G}%
_{\alpha \beta }\}S_{\beta }^{2}Q_{j}-\mathbf{E}\{(h_{n}-\overline{h}%
_{n})H_{\alpha \beta }^{(j)}\}Q_{j}  \label{eh2} \\
&&-n^{-2}\mathbf{E}\{(\mathsf{G}S^{2}\mathsf{G}SX^{T})_{\alpha
j}(XS)_{j\beta }\}\mathsf{K}_{j}Q_{j},  \notag
\end{eqnarray}%
where
\begin{equation}
Q_{j}(z)=\mathsf{K}_{j}(\overline{h}_{n}(z)\mathsf{K}_{j}+1)^{-1}  \label{Qj}
\end{equation}%
is well defined for $z=-\xi <0$. Indeed, since $R$ is positive definite, it
follows from (\ref{hR}), the spectral theorem and (\ref{r2}) that $\overline{%
h}_{n}$ admits the representation%
\begin{equation}
h_{n}(z)=\int_{0}^{\infty }\frac{\mu (d\lambda )}{\lambda -z},\;\mu \geq
0,\;\mu (\mathbb{R}_{+})=n^{-1}\mathrm{Tr}R_{n}\leq r_{2}^{1/2}<\infty .
\label{hn}
\end{equation}%
Thus, we have in view of (\ref{imin})

\begin{equation}
0<h_{n}(-\xi )\leq r_{2}^{1/2}/\xi <\infty ,  \label{hbound}
\end{equation}%
and then the positivity of $K_{j}$ of and (\ref{snkn}) imply
\begin{equation}
0<Q_{j}(-\xi )\leq \mathsf{K}_{j},\;\xi >0.  \label{qjb}
\end{equation}%
We then sum (\ref{eh2}) over $j=1,...,n$ and denote
\begin{equation}
H_{\alpha \beta }=\sum_{j=1}^{n}H_{\alpha \beta }^{(j)},\;H=\{H_{\alpha
\beta }\}_{\alpha ,\beta =1}^{n}\;  \label{hab}
\end{equation}%
yielding
\begin{equation}
\mathbf{E}\{H\}=\mathbf{E}\{\mathsf{G}\}k_{n}(-\xi )R-T,\;T=T^{(1)}+T^{(2)},
\label{eh3}
\end{equation}%
where%
\begin{equation}
k_{n}(-\xi ):=n^{-1}\sum_{j=1}^{n}Q_{j}=\int \frac{\lambda \nu _{\mathsf{K}%
_{n}}(d\lambda )}{\overline{h}_{n}(-\xi )\lambda +1},  \label{kn}
\end{equation}%
$\nu _{\mathsf{K}_{n}}$ is the NCM of $\mathsf{K}_{n} $ (see (\ref{kur1}))
and%
\begin{eqnarray}
T^{(1)} &=&n^{-1}\mathbf{E}\{(h_{n}-\overline{h}_{n})\mathsf{G}SX^{T}\mathsf{%
K}QXS\},  \label{T1} \\
T^{(2)} &=&-n^{-2}\mathbf{E}\{\mathsf{G}S^{2}\mathsf{G}SX^{T}\mathsf{K}QXS\}.
\notag
\end{eqnarray}%
Plugging now (\ref{eh3}) into (\ref{hgm}) and the obtained expression in the
r.h.s. of expectation of (\ref{rid}), we get%
\begin{equation}
\mathbf{E}\{\mathsf{G}\}(k_{n}(z)R-z)=1-T.  \label{EG1}
\end{equation}%
The matrix $(k_{n}(-\xi )R+\xi )$ is invertible uniformly in $n$. Indeed,
since $R$ is positive definite and $k_{n}(-\xi ),\;\xi >0$ is positive in
view of (\ref{hbound}) and (\ref{kn}), we have uniformly in $n\rightarrow
\infty $:%
\begin{equation}
||(k_{n}(-\xi )R+\xi )||^{2}\geq \xi ^{2}.  \label{nor}
\end{equation}%
Thus, we can write instead of (\ref{EG1})%
\begin{equation}
\mathbf{E}\{\mathsf{G}\}=\overline{\mathsf{G}}-\overline{\mathsf{G}}T,\;%
\overline{\mathsf{G}}=(k_{n}(-\xi)R+\xi)^{-1}  \label{EG2}
\end{equation}%
yielding in view of the spectral theorem for $R_{n}$%
\begin{equation}
f_{\mathsf{M}_n}(-\xi)=\int_{0}^{\infty } \frac{\nu _{R_{n}}(d\lambda )}{%
k_{n}(-\xi)\lambda +\xi}+t_{n}(-\xi),\;t_{n}(-\xi)=-n^{-1}\mathrm{Tr}\overline{%
\mathsf{G}}T,  \label{fnpf}
\end{equation}%
where $\nu _{R_{n}}$ is the NCM of $R_{n}$ (see (\ref{nur})).

Next, multiplying (\ref{EG2}) by $R$\ and applying to the result the
operation $n^{-1}\mathrm{Tr}$, we obtain in view of (\ref{hR}) and (\ref{hnb}%
)
\begin{equation}
\overline{h}_{n}(-\xi)=\int_{0}^{\infty }\frac{\lambda \nu _{R_{n}}(d\lambda
)}{k_{n}(-\xi)\lambda +\xi}+\widetilde{t}_{n}(-\xi),\;\widetilde{t}%
_{n}(-\xi)=-n^{-1}\mathrm{Tr}\overline{\mathsf{G}}RT.  \label{hnpf}
\end{equation}%
The integral terms in the r.h.s. of (\ref{kn}), (\ref{fnpf}) and (\ref{hnpf}%
) are obviously the prelimit versions of the r.h.s. of (\ref{fh}) -- (\ref%
{kh}). Thus we have to show that the remainder terms $t_{n}$ and $\widetilde{%
t}_{n}$ in (\ref{fnpf}) and (\ref{hnpf})\ \ vanish as $n\rightarrow \infty $
under the condition (\ref{imin}) and to carry out the limiting transition in
the integral terms of (\ref{kn}), (\ref{fnpf}) and (\ref{hnpf}). The second
procedure is quite standard in random matrix theory and based on (\ref{nur})
and (\ref{kur1}), the compactness of sequences of bounded analytic functions
with respect to the uniform convergence on a compact set of complex plane,
the compactness of sequences on probability measures with respect to the
weak convergence and the unique solvability of the system (\ref{hk}) -- (\ref%
{kh}) proved in Lemma \ref{l:uniq} \ (see, e.g. \cite{Pa-Sh:11} for a number
of examples of the procedure).

Thus, we will deal with the remainders \ in (\ref{fnpf}) -- (\ref{hnpf}). We
will assume for time being that the matrix $R_{n}=S_{n}^{2}$ of (\ref{r2})
is uniformly bounded in $n$ (see (\ref{rrho})). This assumption can be
removed at the end of the proof by using an argument analogous to that used
at the end of proof \ of Theorem \ref{t:ind}. Recall that we are assuming
that $z=-\xi \in I_{-}$ of (\ref{imin}).

We will start with the contribution
\begin{equation}
t_{n}^{(1)}=-n^{-2}\mathbf{E}\{(h_{n}-\overline{h}_{n})\mathrm{Tr}S\mathsf{G}%
\overline{\mathsf{G}}S\mathsf{B}\},\;\mathsf{B}=X^{T}\mathsf{K}QX\geq 0,
\label{tB}
\end{equation}%
of $T^{(1)}$ in (\ref{T1}) to $t_{n}(-\xi )$ of (\ref{fnpf}). We have from (%
\ref{tab}), (\ref{rrho}), (\ref{kgbo}) and (\ref{tB}):%
\begin{eqnarray*}
n^{-2}|\mathrm{Tr}S\mathsf{G}\overline{\mathsf{G}}S\mathsf{B}| &\leq &\rho
(\xi n)^{-2}\mathrm{Tr}\mathsf{B} \\
&\leq &\rho (\xi n)^{-2}\sum_{j=1}^{n}||X^{(j)}||^{2}\mathsf{K}_{j}Q_{j},
\end{eqnarray*}%
where $X^{(j)}$ is the $j$th column of $X$. This, Schwarz inequality for
expectations, (\ref{exj4}), (\ref{k21})\ and (\ref{qjb}) yield%
\begin{eqnarray*}
|t_{n}^{(1)}| &\leq &\rho \xi ^{-2}n^{-2}\sum_{j=1}^{n}\mathsf{K}_{j}Q_{j}%
\mathbf{E}^{1/2}\{||X^{(j)}||^{4}\}\mathbf{E}^{1/2}\{|h(-\xi )-\overline{h}%
_{n}(-\xi )|^{2}\} \\
&\leq &\rho k_{2}\xi ^{-2}C^{1/2}\mathbf{E}^{1/2}\{|h_{n}(-\xi )-\overline{h}%
_{n}(-\xi )|^{2}\}
\end{eqnarray*}%
and then an analog of Lemma \ref{l:mart} (iii) for $\mathsf{M}_{n}$ and (\ref%
{hbound}) implies for every $\xi >0$
\begin{equation}
|t_{n}^{(1)}|=O(n^{-1/2}),\;n\rightarrow \infty .  \label{t2}
\end{equation}%
Similarly, we have for the contribution
\[t_{n}^{(2)}=n^{-3}\mathbf{E}\{%
\mathrm{Tr}S\overline{\mathsf{G}}\mathsf{G}S^{2}\mathsf{G}S\mathsf{B}\}
\] of $%
T^{(2)}$ of (\ref{T1}) to $t_{n}$ in (\ref{hnpf}) by (\ref{tab}), (\ref{rrho}%
) and (\ref{kgbo}): $|t_{n}^{(2)}|\leq \rho ^{2}\xi ^{-3}n^{-3}\mathbf{E}\{%
\mathrm{Tr}\mathsf{B}\}$ and then for every $\xi >0$%
\begin{equation}
t_{n}^{(2)}=O(n^{-1}),\;n\rightarrow \infty .  \label{t3}
\end{equation}%
Combining now (\ref{t2}) -- (\ref{t3}), we obtain $t_{n}(-\xi
)=O(n^{-1/2}),\;n\rightarrow \infty ,\;\xi >0$.

By using a similar argument, we find that $\widetilde{t}_{n}(-\xi
)=O(n^{-1/2}),\;n\rightarrow \infty ,\;\xi >0$. This and (\ref{fnpf}) -- (%
\ref{hnpf}) with $z=-\xi <0$ lead to (\ref{fh}) and (\ref{hk}). Multiplying (%
\ref{hk}) by $k$ and using the first equality in (\ref{fh}), we obtain the
second equality.

The unique solvability of system (\ref{hk}) --\ (\ref{kh}) is proved in
Lemma \ref{l:uniq}. 
\end{proof}

It is convenient to write the equations (\ref{hk}) -- (\ref{kh}) in a
compact form similar to that of free probability theory \cite%
{Ch-Co:18,Mi-Sp:17}. This, in particular, makes
explicit the symmetry and the transitivity of the binary
operation (\ref{opdia}).

\begin{corollary}
\label{c:conv} Let $\nu _{\mathsf{K}},\;\nu _{R}$ and $\nu _{\mathsf{M}}$ be
the probability measures (i.e., non-negative measures of the total mass 1)
entering (\ref{fh}) -- (\ref{kh}) and $m_{\mathsf{K}},\;m_{R}$ and $m_{%
\mathsf{M}}$ be their moment generating functions (see (\ref{mgen}) -- (\ref%
{stmg})). Then the functional inverses $z_{\mathsf{M}},\;z_{\mathsf{K}}$ and
$z_{R}$ of the corresponding moment generating functions are related as
follows
\begin{equation}
z_{\mathsf{M}}(m)=z_{\mathsf{K}}(m)z_{R}(m)m^{-1},  \label{mconv}
\end{equation}%
or, writing $z_{A}(m)=m\sigma _{A}(m),\;A=\mathsf{M},\mathsf{K},R$,%
\begin{equation}
\sigma _{\mathsf{M}}(m)=\sigma _{\mathsf{K}}(m)\sigma _{R}(m)  \label{sconv}
\end{equation}
\end{corollary}

\begin{proof}
It follows from (\ref{hk}) -- (\ref{kh}) and (\ref{stmg}) that%
\begin{eqnarray}
m_{\mathsf{K}}(-h(z)) &=&-h(z)k(z),\;m_{R}(k(z)z^{-1})=-h(z)k(z),
\label{rels} \\
m_{\mathsf{M}}(z^{-1}) &=&-h(z)k(z).  \notag
\end{eqnarray}%
Now the first and the third relations (\ref{rels}) yield $m_{\mathsf{K}%
}(-h(z^{-1}))=m_{\mathsf{M}}(z)$, hence $z_{\mathsf{K}}(m)=-h(z_{\mathsf{M}%
}^{-1}(m))$, and then the second and the third relations yield $%
m_{R}(k(z^{-1})z)=m_{\mathsf{M}}(z)$, hence $z_{R}(\mu )=k(z_{\mathsf{M}%
}^{-1}(m))z_{\mathsf{M}}(m)$. Multiplying these two relations and using once
more the third relation in (\ref{rels}), we obtain%
\begin{equation*}
z_{\mathsf{K}}(m)z_{R}(m)=-k(z_{\mathsf{M}}^{-1}(m))h(z_{\mathsf{M}%
}^{-1}(m))z_{\mathsf{M}}(m)=z_{\mathsf{M}}(m)m
\end{equation*}%
and (\ref{mconv}) and (\ref{sconv}) follows.
\end{proof}

\begin{remark}
\label{r:diamond}  In the case of rectangular matrices $X_{n}$ in (\ref%
{mncal}), described in Remark \ref{r:rra1} (ii), the analogs of (\ref{mconv}%
) and (\ref{sconv}) are
\begin{equation}
z_{\mathsf{M}}(m)=z_{\mathsf{M}}(cm)z_{R}(cm)m^{-1},\;\sigma _{\mathsf{M}%
}(m)=c^{2}\sigma _{\mathsf{K}}(cm)\sigma _{R}(cm).  \label{rect}
\end{equation}
\end{remark}


\begin{lemma}
\bigskip \label{l:mart}Let $\mathcal{M}_{n}$ be given by (\ref{mncal}) in
which $X_{n}=\{X_{j\alpha }\}_{j,\alpha =1}^{n}$ of (\ref{Xn}) and $%
b_{n}=\{b_{j\alpha }\}_{j}^{n}$ of (\ref{b}) are i.i.d. random variables.
Denote $\nu _{\mathcal{M}_{n}}$ the Normalized Counting Measure of $\mathcal{%
M}_{n}$ (see, e.g. (\ref{ncm})), \ $g_{n}(z)$ its Stieltjes transform
\begin{equation*}
g_{n}(z)=n^{-1}\mathrm{Tr}(\mathcal{M}_{n}-z)^{-1},\;\zeta =\mathrm{dist}(z,%
\mathbb{R}_{+})>0,
\end{equation*}%
and (see (\ref{hR}))%
\begin{equation*}
h_{n}(z)=n^{-1}\mathrm{Tr}S_{n}(\mathcal{M}_{n}-z)^{-1}S_{n},\;\zeta =%
\mathrm{dist}(z,\mathbb{R}_{+})>0,
\end{equation*}%
where $S_{n}$ is a positive definite matrix satisfying (\ref{r2}) with $%
R_{n}=S_{n}^{2}$. Then we have:

(i) for any $n$-independent interval $\Delta $ of spectral axis
\begin{equation*}
\mathbf{E}\{|\nu _{\mathcal{M}_{n}}(\Delta )-\mathbf{E}\{\nu _{\mathcal{M}%
_{n}}(\Delta )\}|^{4}\}\leq C_{1}/n^{2},
\end{equation*}%
where $C_{1}$ is an absolute constant;

(ii) for any $n$-independent $z$ with $\zeta >0$%
\begin{equation*}
\mathbf{E}\{|g_{n}-\mathbf{E}\{g_{n}\}|^{4}\}\leq C_{2}/n^{2}\zeta ^{4},
\end{equation*}%
where $C_{2}$ is an absolute constant;

(iii) for any $n$-independent $z$ with $\zeta>0$%
\begin{equation*}
\mathbf{E}\{|h_{n}-\mathbf{E}\{h_{n}\}|^{4}\}\leq C_{3}r_{2}^{2}/n^{2}\zeta
^{4},
\end{equation*}%
where $C_{3}$ is an absolute constant and $r_{2}$ is defined in (\ref{r2}).
\end{lemma}

\begin{proof}
\bigskip It follows from (\ref{mncal}) that (cf. (\ref{nmyy}))
\begin{equation}
\mathcal{M}_{n}=\sum_{j=1}^{n}\mathcal{Y}_{j}\otimes \mathcal{Y}_{j},\;%
\mathcal{Y}_{j}=\{\mathcal{Y}_{j\alpha }\}_{\alpha =1}^{n},\;\mathcal{Y}%
_{j\alpha }=n^{-1/2}(D_{n}XS_{n})_{j\alpha }.  \label{cmyy}
\end{equation}%
It is easy to see that $\{\mathcal{Y}_{j}\}_{j=1}^{n}$ are independent. This
allows us to use the martingale bounds given in Sections 18.2 and 19.1 of
\cite{Pa-Sh:11} and implying the assertions of the lemma in view of (\ref{Gb}%
) and (\ref{hn}) .
\end{proof}

\begin{remark}\label{r:mxind}
(i) The independence of random vectors $\mathcal{Y}_j$ in (\ref{cmyy}) is
the main reason to pass from the matrices $M_n^l$ given by (\ref{m1}) and (\ref%
{m2m1}) to the matrices $\mathcal{M}_n$ given by (\ref{cm1}), (\ref{cm1m2}) and
(\ref{mncal}).

(ii) The lemma is valid for an arbitrary (not necessarily Gaussian)
collection (\ref{Xn}) and (\ref{b}) of i.i.d. random variables as well as
for random but independent of (\ref{Xn}) and (\ref{b}) $S_{n}$ and $%
\{x_{j,\alpha }\}_{j,\alpha =1}^{n}$, see Remarks \ref{r:rra} (i) and (iii)
and \cite{Pa-Sl:20}. It is also valid for matrices $\mathsf{M}_{n}$ of (\ref{mnbf}).
\end{remark}


The next lemma deals with asymptotic properties of the vectors of
activations $x^{l}$ in the $l$th layer, see (\ref{rec}). It is an extended
version (treating the convergence with probability 1) of assertions proved
in \cite{Ma-Co:16,Po-Co:16,Sc-Co:17}.


\begin{lemma}
\label{l:xlyl} Let $y^{l}=\{y_{j}^{l}\}_{j=1}^{n},\;l=1,2,...$ be
post-affine random vectors defined in (\ref{rec}) -- (\ref{wga}) with $x^{0}$
satisfying (\ref{q0}), $\chi :\mathbb{R}\rightarrow \mathbb{R}$ be a bounded
piecewise continuous function and $\Omega _{l}$ be defined in (\ref{oml}).
Set
\begin{equation}
\chi _{n}^{l}=n^{-1}\sum_{j_{l}=1}^{n}\chi (y_{j_{l}}^{l}),\;l \ge 1.
\label{chiln}
\end{equation}%
Then there exists $\overline{\Omega }_{l}\subset \Omega _{l},\;\mathbf{P}(%
\overline{\Omega }_{l})=1$ such that for every $\omega _{l}\in \overline{%
\Omega }_{l}$ (i.e., with probability 1) the limit%
\begin{equation}
\chi ^{l}:=\lim_{n\rightarrow \infty }\chi _{n}^{l},\;l=1,2,...,  \label{lql}
\end{equation}%
exists, is non-random and given by the formula
\begin{equation}
\chi ^{l}=\int_{-\infty }^{\infty }\chi (\gamma \sqrt{q^{l}})\Gamma (d\gamma
),\;l=1,2,...,  \label{lqg}
\end{equation}%
valid on $\overline{\Omega }_{l}$ with\ $\Gamma (d\gamma )=(2\pi
)^{-1/2}e^{-\gamma ^{2}/2}d\gamma $ being the standard Gaussian probability
distribution and $q^{l}$ defined recursively by the formula
\begin{equation}
q^{l}=\int_{-\infty }^{\infty }\varphi ^{2}(\gamma \sqrt{q^{l-1}})\Gamma
(d\gamma )+\sigma _{b}^{2},\;l=2,3,...  \label{ql}
\end{equation}%
and by $q_{1}$ of (\ref{q0}). \

In particular, we have with probability 1:

(i) for the activation vector $x^{l}=\{x_{j}^{l}\}_{j=1}^{n}$ of the $l$th
layer (see (\ref{rec})):%
\begin{equation}
\lim_{n\rightarrow \infty
}n^{-1}\sum_{j_{l}=1}^{n}(x_{j_{l}}^{l})^{2}=q^{l+1}-\sigma
_{b}^{2},\;l=1,2,...,  \label{xlim}
\end{equation}

(ii) for the weak limit $\nu _{K^{l}}$ of the Normalized Counting Measure $%
\nu _{K_{n}^{l}}$ of diagonal random matrix $K_{n}^{l}$ of (\ref{kan}): $\nu
_{K^{l}}$ is the probability distribution of the random variable $(\varphi
^{\prime }(\gamma \sqrt{q_{l}}))^{2}$.
\end{lemma}

\begin{proof}
Set $l=1$ in (\ref{chiln}) Since $\{b_{j_{1}}^{1}\}_{j_{1}=1}^{n}$ and $%
\{X_{j_{1}}^{1}\}_{j_{1}=1}^{n}$ are i.i.d. Gaussian random variables
satisfying (\ref{bga}) -- (\ref{wga}), it follows from (\ref{rec}) that the
components of $y^{1}=\{y_{j_{1}}^{1}\}_{j_{1}=1}^{n}$ are also i.i.d.
Gaussian random variables of zero mean and variance $q_{n}^{1}$ of (\ref{q0}%
). Since $\chi $ is bounded, the collection $\{\chi
(y_{j_{1}}^{1})\}_{j_{1}=1}^{n}$ consists of bounded i.i.d random variables
defined for all $n$ on the same probability space $\Omega _{1}$ generated by
(\ref{xinf}) and (\ref{binf}) with $l=1$. This allows us to apply to $\{\chi
(y_{j_{1}}^{1})\}_{j_{1}}^{n}$ the strong Law of Large Numbers implying (\ref%
{lql}) with $l=1$ together with the formula%
\begin{eqnarray}
\chi ^{1} &=&\lim_{n\rightarrow \infty }\mathbf{E}\{\chi (y_{1}^{1})\}
\label{chi10} \\
&=&\lim_{n\rightarrow \infty }\int_{-\infty }^{\infty }\chi (\gamma \sqrt{%
q_{n}^{1}})\Gamma (d\gamma )=\int_{-\infty }^{\infty }\chi (\gamma \sqrt{%
q^{1}})\Gamma (d\gamma )  \notag
\end{eqnarray}%
for the limit, both valid with probability 1, i.e., on a certain $\overline{%
\Omega }_{1}\subset \Omega _{1}=\Omega ^{1}$, $\mathbf{P}(\Omega _{1})=1$,
see (\ref{oml}). This yields (\ref{lqg}) for $l=1$.

Consider now the case $l=2$. Since $\{X^{1},b^{1}\}$ and $\{X^{2},b^{2}\}$
are independent collections of random variables, we can fix $\omega _{1}\in
\overline{\Omega }_{1}$ (a realization of $\{X^{1},b^{1}\}$) and apply to $%
\chi _{n}^{2}$ of (\ref{chiln}) the same argument as that for the case $l=1$
above to prove that for every $\omega _{1}\in \overline{\Omega }_{1}$ there
exists $\overline{\Omega }^{2}(\omega ^{1})\subset \Omega ^{2},\;\mathbf{P}(%
\overline{\Omega }^{2})=1\,$\ on which we have (\ref{lql}) for $l=2$ with
some (cf. (\ref{chi10}))
\begin{equation}
\chi ^{2}(\omega ^{1},\omega ^{2})=\lim_{n\rightarrow \infty }\mathbf{E}%
_{\{X^{2},b^{2}\}}\{\chi (y_{1}^{2})\}  \label{q2oo}
\end{equation}%
%
where $\mathbf{E}_{\{X^{2},b^{2}\}}\{...\}$ denotes the expectation with
respect to $\{X^{2},b^{2}\}$ only. Now the Fubini theorem implies that there
exists $\overline{\Omega }_{2}\subset \Omega _{2}=\Omega ^{1}\otimes \Omega
^{2},\;$ $\mathbf{P}(\overline{\Omega }_{2})=1$ on which we have (\ref{lql})
with $l=2$.

Using once more the independence of $\{X^{1},b^{1}\}$ and $\{X^{2},b^{2}\},$
we can compute the r.h.s. of (\ref{q2oo}) by observing that if $%
\{X^{2},b^{2}\}$ are Gaussian, then, according to (\ref{rec}), $y_{1}^{2}$
is also Gaussian of zero mean and variance (cf. (\ref{q0}))%
\begin{equation}
q_{n}^{2}=n^{-1}\sum_{j_{1}=1}^{n}(x_{j_{1}}^{1})^{2}+\sigma _{b}^{2},
\label{q2n}
\end{equation}%
or, in view of (\ref{rec}),%
\begin{equation}
q_{n}^{2}=n^{-1}\sum_{j_{1}=1}^{n}(\varphi (y_{j_{1}}^{1}))^{2}+\sigma
_{b}^{2}.  \label{q2nf}
\end{equation}%
The first term on the right is a particular case of (\ref{chiln}) with $\chi
=(\varphi) ^{2}$ and $l=1$, thus, according to (\ref{chi10}), the limiting
form of the above relation is (\ref{ql}) with $l=2$ for every $\omega
_{1}\in \overline{\Omega }_{1}$ and we have
\begin{equation*}
\chi ^{2}=\lim_{n\rightarrow \infty }\int_{-\infty }^{\infty }\chi (\gamma
\sqrt{q_{n}^{2}})\Gamma (d\gamma )=\int_{-\infty }^{\infty }\chi (\gamma
\sqrt{q^{2}})\Gamma (d\gamma )
\end{equation*}%
i.e., formula (\ref{lqg}) for $l=2$ valid on $\overline{\Omega }_{2}\subset
\Omega _{2}=\Omega ^{1}\otimes \Omega ^{2},\;$ $\mathbf{P}(\overline{\Omega }%
_{2})=1$, i.e., with probability 1.

This proves the validity (\ref{lql}) -- (\ref{ql}) for $l=2$ with
probability 1. Analogous argument applies for $l=3,4,....$

The proof of item (i) is, in fact, that of (\ref{ql}), see (\ref{q2n}) -- (%
\ref{q2nf}) for $l=2$, for $l\geq 3$ the proof is analogous.

Let us prove item (ii) of the lemma, i.e., the weak convergence with
probability 1 of the Normalized Counting Measure $\nu _{K_{n}^{l}}$ of $%
K_{n}^{l}$ in (\ref{kan}) to the probability distribution of $(\varphi
^{\prime }(\gamma \sqrt{q_{l}}))^{2}$. It suffices to prove the validity
with probability 1 of the relation
\begin{equation*}
\lim_{n\rightarrow \infty }\int_{-\infty }^{\infty }\psi (\lambda )\nu
_{K_{n}^{l}}(d\lambda )=\int_{-\infty }^{\infty }\psi (\lambda )\nu
_{K^{l}}(d\lambda )
\end{equation*}%
for any bounded and piece-wise continuous $\psi :\mathbb{R\rightarrow R}$.

In view of (\ref{rec}), (\ref{D}) and (\ref{kan}) the relation can be
written in the form%
\begin{equation*}
\lim_{n\rightarrow \infty }n^{-1}\sum_{j_{l}=1}^{n}\psi ((\varphi ^{\prime
}(y_{j_{l}}^{l}))^{2})=\int_{-\infty }^{\infty }\psi ((\varphi ^{^{\prime
}}(\gamma \sqrt{q_{l-1}}))^{2})\Gamma (d\gamma ),\;l\ge 1.
\end{equation*}%
This is a particular case of (\ref{lql}) -- (\ref{ql}) for $\chi =\psi \circ
\varphi ^{\prime 2}$, hence, assertion (ii) follows.
\end{proof}


\medskip The next lemma provides the unique solvability of the system (\ref%
{hk}) -- (\ref{kh}). Note that in the course of proving Lemma \ref{l:rkx} it
was proved that the system has at least one solution.

\begin{lemma}
\bigskip \label{l:uniq} The system (\ref{hk}) -- (\ref{kh}) with $\nu _{R}$
and $\ \nu _{K}$ satisfying
\begin{equation}
\nu _{K}(\mathbb{R}_{+})=1,\;\nu _{R}(\mathbb{R}_{+})=1  \label{mkr}
\end{equation}%
$\;$ and (cf. ((\ref{r2}))%
\begin{equation}
\int_{0}^{\infty }\lambda ^{2}\nu _{K}(d\lambda )=\kappa _{2}<\infty
,\;\int_{0}^{\infty }\lambda ^{2}\nu _{R}(d\lambda )=\rho _{2}<\infty
\label{nukr}
\end{equation}%
has a unique solution in the class of pairs of functions $(h,k)$ defined in $%
\mathbb{C}\setminus \mathbb{R}_{+}$ and such that $h$ is analytic in $%
\mathbb{C}\setminus \mathbb{R}_{+}$, continuous and positive on the open
negative semi-axis and satisfies (\ref{hcond}) with $r_{2}$ replaced by $%
\rho _{2}$ of (\ref{nukr}).

\smallskip
Besides:

\smallskip
(i) the function $k$ is analytic in $\mathbb{C}\setminus \mathbb{R}_{+}$,
continuous and positive on the open negative semi-axis and (cf. (\ref{hcond}%
))%
\begin{equation}
\Im k(z)\Im z<0\;\mathrm{for\;}\Im z\neq 0,\;0<k(-\xi )\leq \kappa
_{2}^{1/2}\;\mathrm{for\;}\xi >0  \label{imk}
\end{equation}%
with $\kappa _{2}$ of (\ref{nukr});

(ii) if the sequence $\{\nu _{K^{(p)}},\nu _{R^{(p)}}\}_{p}$ has uniformly
in $p$ bounded second moments (see (\ref{nukr})) and converges weakly to $%
(\nu _{K},\nu _{R})$ also satisfying (\ref{nukr}), then the sequences of the
corresponding solutions $\{h^{(p)},k^{(p)}\}_{p}$ of the system (\ref{hk})
-- (\ref{kh}) converges pointwise in $\ \mathbb{C}\setminus \mathbb{R}_{+}$
to the solution $(h,k)$ of the system corresponding to the limiting measures
$(\nu _{K},\nu _{R})$.
\end{lemma}

\begin{proof}
We will start with the proof of assertion (i). It follows from (\ref{kh}), (%
\ref{nukr}) and the analyticity of $h$ in $\mathbb{C}\setminus \mathbb{R}%
_{+} $ that $k$ is also analytic in $\mathbb{C}\setminus \mathbb{R}_{+}$.
Next, for any solution of (\ref{hk}) -- (\ref{kh}) we have \ from (\ref{kh})
with $\Im z\neq 0$%
\begin{equation}
\Im k(z)=-\Im h(z)\int_{0}^{\infty }\frac{\lambda ^{2}\nu _{K}(d\lambda )}{%
|h(z)\lambda +1|^{2}}  \label{imkh}
\end{equation}%
and then (\ref{hcond}) yields (\ref{imk}) for $\Im z\neq 0$, while (\ref{kh})
with $z=-\xi <0$
\begin{equation*}
k(-\xi )=\int_{0}^{\infty }\frac{\lambda \nu _{K}(d\lambda )}{h(-\xi
)\lambda +1},
\end{equation*}%
the positivity of $h(-\xi )$ (see (\ref{hcond})), (\ref{mkr}) and Schwarz
inequality yield (\ref{imk}) for $z=-\xi $.

Let us prove now that the system (\ref{hk}) -- (\ref{kh}) is uniquely
solvable in the class of pairs of functions $(h,k)$ analytic in $\mathbb{C}%
\setminus \mathbb{R}_{+}$ and satisfying (\ref{hcond}) and (\ref{imk}).

Denote $\mathbb{C}_{+}$ and $\mathbb{C}_{-}$ the upper and lower open
half-planes. Consider first the case $z\in \mathbb{C}_{+}$ of the system (%
\ref{hk}) -- (\ref{kh}). To this end introduce the map%
\begin{equation}
F:\{h\in \mathbb{C}_{+}\}\times \{k\in \mathbb{C}_{-}\}\times \{z\in \mathbb{%
C}_{+}\}\rightarrow \mathbb{C\times C}  \label{F}
\end{equation}%
defined by%
\begin{align}
& F_{1}(h,k,z)=h-\int_{0}^{\infty }\frac{\lambda \nu _{R}(d\lambda )}{%
k\lambda -z},h  \label{f12} \\
& F_{2}(h,k,z)=k-\int_{0}^{\infty }\frac{\lambda \nu _{K}(d\lambda )}{%
h\lambda +1}.  \notag
\end{align}%
The map is well defined in the indicated domain, since there $\Im |k\lambda
-z|>\lambda |\Im k|$ and $\Im |1+h\lambda |>\lambda \Im h$, hence the
absolute values of the integrals in $F_{1}$ and that in $F_{2}$ are bounded
from above by $|\Im k|^{-1}<\infty $ and $(\Im h)^{-1}<\infty $
respectively. The equation%
\begin{equation}
F(h,k,z)=0  \label{feq}
\end{equation}%
is in fact (\ref{hk}) -- (\ref{kh}). We will apply now to the equation the
implicit function theorem. To this end we have to prove that the Jacobian of
$F$, i.e., $2\times 2$ matrix of derivatives of $F$ with respect to $h$ and $%
k$, is invertible. It is easy to find that the determinant of the Jacobian
is
\begin{equation}
1-I(h)J(k,z)  \label{detj}
\end{equation}%
with%
\begin{equation*}
I(h)=\int_{0}^{\infty }\frac{\lambda ^{2}\nu _{K}(d\lambda )}{(h\lambda
+1)^{2}}\leq A(h),\;J(k,z)=\int_{0}^{\infty }\frac{\lambda ^{2}\nu
_{R}(d\lambda )}{(k\lambda -z)^{2}}\leq B(k,z)
\end{equation*}%
and
\begin{eqnarray*}
0 &<&A(h):=\int_{0}^{\infty }\frac{\lambda ^{2}\nu _{K}(d\lambda )}{%
|h\lambda +1|^{2}}\leq (\Im h)^{-2}<\infty , \\
0 &<&B(k,z):=\int_{0}^{\infty }\frac{\lambda ^{2}\nu _{R}(d\lambda )}{%
|k\lambda -z|^{2}}\leq (\Im k)^{-2}<\infty ,
\end{eqnarray*}%
where we used (\ref{mkr}) to obtain the second inequality.

On the other hand, the imaginary part of (\ref{f12}) -- (\ref{feq}) yield
(cf. (\ref{imkh}))
\begin{equation*}
A(h)=-\frac{\Im k}{\Im h},\;B(k,z)=-\frac{\Im h}{\Im k}+\frac{\Im z}{\Im k}%
C(k,z),
\end{equation*}%
where%
\begin{equation}
0<C(k,z)=\int_{0}^{\infty }\frac{\lambda \nu _{R}(d\lambda )}{|k\lambda
-z|^{2}}<\infty ,\;\Im z>0.  \label{C}
\end{equation}%
This implies%
\begin{equation}
0<A(h)B(k,z)=1-\Im z(\Im h)^{-1}C(k,z)  \label{ABC}
\end{equation}%
and since $C(k,z)\,\Im z(\Im h)^{-1}>0$ in view of (\ref{F}) and (\ref{C}),
we have for the determinant (\ref{detj})%
\begin{equation}
|1-I(h)J(k,z)|\geq 1-A(h)B(k,z)=C(k,z)\,\Im z(\Im h)^{-1}>0.  \label{jacp}
\end{equation}%
Thus, the Jacobian of the map (\ref{F}) -- (\ref{f12}) is invertible and the system (\ref%
{hk}) \ -- (\ref{kh}) is uniquely solvable in $\mathbb{C}_{+}$. The proof
for $\mathbb{C}_{-}$ is analogous.

Assume now that $z=-\xi ,\;\xi >0$. Here we consider the map%
\begin{equation*}
\widetilde{F}:\{h\in \mathbb{R}_{+}\setminus \{0\}\}\times \{k\in \mathbb{R}%
_{+}\}\times \{z=-\xi \in \mathbb{R}_{-}\setminus \{0\}\}\rightarrow \mathbb{%
R\times R}
\end{equation*}%
defined by (\ref{f12}) with $h>0,\;k>0,\;z=-\xi <0$. It is easy to find that
the map is well defined since $k\lambda +\xi >k\lambda \geq 0,\;1+h\lambda
>h\lambda \geq 0$, hence the integrals in (\ref{F}) with $h>0,\;k>0,\;z=-\xi
<0$ are positive and bounded from above by $k^{-1}<\infty $ and $%
h^{-1}<\infty $ respectively. Moreover, since in this case we have
\begin{eqnarray*}
I(h) &=&A(h)=\int_{0}^{\infty }\frac{\lambda ^{2}\nu _{K}(d\lambda )}{%
(h\lambda +1)^{2}},^{{}} \\
J(k,-\xi ) &=&B(k,-\xi )=\int_{0}^{\infty }\frac{\lambda ^{2}\nu
_{R}(d\lambda )}{(k\lambda +\xi )^{2}},
\end{eqnarray*}%
the determinant of the Jacobian of $\widetilde{F}$ is now (cf. (\ref{detj}))
\begin{equation*}
1-A(h)B(k,-\xi ),\;\;h>0,\;k>0,\;\xi >0.
\end{equation*}%
Set in (\ref{ABC}) $h=h^{\prime }+i\varepsilon ,\;k=k^{\prime }-i\varepsilon
,\;z=-\xi +i\varepsilon $ where $h^{\prime }>0,\;k^{\prime }>0,\;\xi
>0,\;\varepsilon >0$ and carry out the limit $\varepsilon \rightarrow 0$. We
obtain (cf. (\ref{jacp}))
\begin{equation*}
1-A(h^{\prime })B(k^{\prime },-\xi )=C(k^{\prime },-\xi )=\int_{0}^{\infty }%
\frac{\lambda \nu _{R}(d\lambda )}{(k\lambda +\xi )^{2}}>0.
\end{equation*}%
This proves the unique solvability of (\ref{hk}) -- (\ref{kh}) in $\mathbb{C}%
\setminus \mathbb{R}_{+}$.

Let us prove assertion (ii) of the lemma. Since $h^{(p)}$ and $k^{(p)}$ are
analytic and uniformly in $p$ bounded outside the closed positive semiaxis,
there exist subsequences $\{h^{(p_{j})},k_{{}}^{(p_{j})}\}_{j}$ converging
pointwise in $\mathbb{C}\setminus \mathbb{R}_{+}$ to a certain analytic pair
$(\widetilde{h},\widetilde{k})$. Let us show that $(\widetilde{h},\widetilde{%
k})=(h,k)$. It suffices to consider real negative $z=-\xi >0$ (see (\ref%
{imin})). Write for the analog of (\ref{kh}) for $\nu _{K^{(p)}}$:%
\begin{eqnarray*}
k^{(p)} &=&\int_{0}^{\infty }\frac{\lambda \nu _{K^{(p)}}(d\lambda )}{%
h^{(p)}\lambda +1} \\
&=&\int_{0}^{\infty }\frac{\lambda \nu _{K^{(p)}}(d\lambda )}{\widetilde{h}%
\lambda +1}+(\widetilde{h}-h^{(p)})\int_{0}^{\infty }\frac{\lambda ^{2}\nu
_{K^{(p)}}(d\lambda )}{(h^{(p)}\lambda +1)(\widetilde{h}\lambda +1)}.
\end{eqnarray*}%
Putting here $p=p_{j}\rightarrow \infty $, we see that the l.h.s. converges
to $\widetilde{k}$, the first integral on the right converges to the r.h.s
of (\ref{kh}) with $\widetilde{h}$ instead of $h$ since $\nu _{K^{(p)}}$
converges weakly to $\nu _{K}$, the integrand is bounded and continuous
and the second integral is bounded in $p$ since $h^{(p)}(-\xi )>0$, $%
\widetilde{h}(-\xi )>0$ and the second moment of $\nu _{K^{(p)}}$ is bounded
in $p$ according to (\ref{nukr}), hence, the second term vanishes as $%
p=p_{j}\rightarrow \infty $. An analogous argument applied to (\ref{hk}%
) show $(\widetilde{h},\widetilde{k})$ is a solution of (\ref{hk}) -- (\ref%
{kh}) and then the unique solvability of the system implies that $(%
\widetilde{h},\widetilde{k})=(h,k)$.
\end{proof}

\medskip 
\textbf{Acknowledgment}. We are grateful to Dr. M. Simbirsky for introducing
us to fascinating field of machine learning and interesting discussions.
We are also grateful to the referee for the careful reading of the manuscript
and for suggestions which helped us to improve considerably our presentation.

\end{document}